\documentclass[11pt]{amsart}
\usepackage{amsmath,amsfonts,amsthm,amscd,amssymb,graphicx}
\numberwithin{equation}{section}

\newtheorem{theorem}{Theorem}[section]

\newtheorem{lemma}[theorem]{Lemma}

\newtheorem{proposition}[theorem]{Proposition}
\newtheorem{prop}[theorem]{Proposition}

\newtheorem{cor}[theorem]{Corollary}
\newtheorem{example}[theorem]{Example}
\newtheorem{remark}[theorem]{Remark}
\newtheorem{definition}[theorem]{Definition}

\def\eps{\varepsilon }
\def\D{\partial }
\newcommand{\na}{{\nabla}}
\renewcommand{\div}{{\rm div}}

\newcommand{\RR}{\mathbb{R}}
\newcommand{\cO}{\mathcal{O}}
\newcommand{\cS}{\mathcal{S}}

\newcommand{\cE}{\mathcal{E}}

\newcommand{\dt}{\frac{d}{dt}}

\newcommand{\CalB}{\mathcal{B}}
\newcommand{\CalF}{\mathcal{F}}

\newcommand{\CalS}{\mathcal{S}}

\newcommand{\const}{\text{\rm constant}}

\newcommand{\e}{{\epsilon}}

\newcommand{\Span}{{\rm span }  \,}

\def\bb1{{1\!\!1}}
\def\bU{{\bar{U}}}
\def\bW{{\bar{W}}}

\def\tx{\tilde x}

\def\R{\Re e}
\def\I{\Im m}

\def\cD{\mathcal{D}}

\newcommand{\iprod}[1]{{\langle{#1}\rangle}_0}
\newcommand{\wprod}[1]{\langle{#1}\rangle}

\newcommand{\txi}{{\tilde \xi}}

\begin{document}

\title [Stability of multi-dimensional boundary layers]
{Long-time stability of multi-dimensional noncharacteristic viscous
boundary layers}

\author[T. Nguyen and K. Zumbrun]{Toan Nguyen and Kevin Zumbrun}

\date{\today}

\thanks{This work was supported in part by the National Science Foundation award number DMS-0300487.}

\address{Department of Mathematics, Indiana University, Bloomington, IN 47402}
\email{nguyentt@indiana.edu}
\address{Department of Mathematics, Indiana University, Bloomington, IN 47402}
\email{kzumbrun@indiana.edu}

\maketitle

\begin{abstract}
We establish long-time stability of multi-dimensional
noncharacteristic boundary layers of a class of
hyperbolic--parabolic systems including the compressible
Navier--Stokes equations with inflow [outflow] boundary conditions,
under the assumption of strong spectral, or uniform Evans,
stability. Evans stabiity has been verified for small-amplitude
layers by Gu\`es, M\'etivier, Williams, and Zumbrun.  For
large-amplitude layers, it may be efficiently checked numerically,
as done in the one-dimensional case by Costanzino, Humpherys,
Nguyen, and Zumbrun.
\end{abstract}

\tableofcontents


\section{Introduction}
We consider a boundary layer, or stationary solution,
\begin{equation}\label{profile}
\tilde U=\bU(x_1), \quad \lim_{z\to +\infty} \bU(z)=U_+, \quad
\bU(0)=\bar U_0
\end{equation}
of a system of conservation laws on the quarter-space
\begin{equation}\label{sys}
\tilde U_t +  \sum_jF^j(\tilde U)_{x_j} = \sum_{jk}(B^{jk}(\tilde
U)\tilde U_{x_k})_{x_j}, \quad x\in \mathbb{R}^{d}_+ =
\{x_1>0\},\quad t>0,
\end{equation}
$\tilde U,F^j\in \mathbb{R}^n$, $B^{jk}\in\mathbb{R}^{n \times n}$,
with initial data $\tilde U(x,0)=\tilde U_0(x)$ and Dirichlet type
boundary conditions specified in \eqref{inBC}, \eqref{outBC} below.
A fundamental question
connected to the physical motivations from aerodynamics
%
is whether or not such boundary layer solutions are {\it stable} in
the sense of PDE, i.e., whether or not a sufficiently small
perturbation of $\bU$ remains close to $\bU$, or converges
time-asymptotically to $\bU$, under the evolution of \eqref{sys}.
That is the question we address here.

\subsection{Equations and assumptions}
We consider the general hyperbolic-parabolic system of conservation
laws \eqref{sys} in conserved variable $\tilde U$, with
$$\tilde U = \begin{pmatrix}\tilde u\\
\tilde v\end{pmatrix}, \quad B=\begin{pmatrix}0 & 0 \\
b^{jk}_1 & b^{jk}_2\end{pmatrix},$$
$\tilde u\in \RR^{n-r}$, and $\tilde v\in \RR^{r}$, where $$
\Re\sigma
\sum_{jk} b_2^{jk}\xi_j\xi_k \ge
\theta |\xi|^2>0, \quad \forall \xi \in \RR^n\backslash \{0\}.$$

Following \cite{MaZ4,Z3,Z4}, we assume that equations \eqref{sys}
can be written,
alternatively, after a triangular change of coordinates
\begin{equation}\label{Wcoord}
\tilde W:=\tilde W(\tilde U) =\begin{pmatrix}\tilde w^I(\tilde u)\\
 \tilde w^{II}(\tilde u, \tilde v)\end{pmatrix},
\end{equation}
in {\em the quasilinear, partially symmetric hyperbolic-parabolic
form}
\begin{equation}\label{symmetric}\tilde A^0 \tilde W_t + \sum_j\tilde A^j\tilde W_{x_j} = \sum_{jk}(\tilde
B^{jk} \tilde W_{x_k})_{x_j} + \tilde G,
\end{equation}
where, defining $\tilde W_+:=\tilde W(U_+)$,
\medskip

(A1) $\tilde A^j(\tilde W_+),\tilde A^0,\tilde A^1_{11}$ are
symmetric, $\tilde A^0$ block diagonal, $\tilde A^0\ge \theta_0>0$,
\medskip

(A2) for each $\xi \in \RR^d\setminus \{0\}$, no eigenvector of
$\sum_j\xi_j\tilde A^j(\tilde A^0)^{-1}(\tilde W_+)$ lies in the
kernel of $\sum_{jk}\xi_j\xi_k\tilde B^{jk}(\tilde A^0)^{-1}(\tilde
W_+)$,
\medskip

(A3) $\tilde B^{jk}=\begin{pmatrix}0 & 0 \\ 0 & \tilde
b^{jk}\end{pmatrix}$, $\sum\tilde b^{jk}\xi_j\xi_k\ge
\theta|\xi|^2$, and $\tilde G=\begin{pmatrix}0\\\tilde g
\end{pmatrix}$ with $\tilde g(\tilde W_x,\tilde W_x)=\cO(|\tilde
W_x|^2).$
\medskip

Along with the above structural assumptions, we make the following
technical hypotheses:
\medskip

(H0) $F^j, B^{jk}, \tilde A^0, \tilde A^j, \tilde B^{jk}, \tilde
W(\cdot), \tilde g(\cdot,\cdot) \in C^{s}$, with 
$s\ge [(d-1)/2]+5$ in
our analysis of linearized stability, and 
$s\ge s(d):=[(d-1)/2]+7$ in
our analysis of nonlinear stability.
\medskip

(H1) $\tilde A_1^{11}$ is either strictly positive or strictly
negative, that is, either $\tilde A_1^{11}\ge \theta_1>0,$ or
$\tilde
A^{11}_1\le -\theta_1<0$. (We shall 
call these cases the {\em inflow case} or {\em outflow case},
correspondingly.)
\medskip

(H2) The eigenvalues of $dF^1(U_+)$ are 
distinct and nonzero.

\medskip

(H3) The eigenvalues of $\sum_j dF^j_\pm \xi_j$ 
have constant multiplicity with respect to $\xi\in \RR^d$, $\xi\ne 0$.

\medbreak

(H4)
The set of branch points of the eigenvalues of
$(\tilde A^1)^{-1}(i\tau \tilde A^0+ \sum_{j\ne 1} i\xi_j\tilde A^j)_\pm$,
$\tau \in \RR$, $\tilde \xi\in \RR^{d-1}$
is the (possibly intersecting) union of
finitely many smooth curves $\tau=\eta_q^\pm(\tilde \xi)$, on which
the branching eigenvalue has constant multiplicity $s_q$ (by
definition $\ge 2$).

\medbreak

\noindent Condition (H1) corresponds to hyperbolic--parabolic
noncharacteristicity, while
(H2) is the condition for the hyperbolicity at $U_+$ of the associated
first-order hyperbolic system obtained by dropping second-order terms. The
assumptions (A1)-(A3) and (H0)-(H2) are satisfied for gas dynamics
and MHD with van der Waals equation of state under inflow or outflow
conditions; see discussions in \cite{MaZ4,CHNZ,GMWZ5,GMWZ6}.
Condition (H3) holds always for gas dynamics, but fails always for
MHD in dimension $d\ge 2$.
Condition (H4) is a technical requirement of the analysis
introduced in \cite{Z2}.  It is satisfied always in dimension $d=2$
or for rotationally invariant systems in dimensions $d\ge 2$, for
which it serves only to define notation; in particular, it holds
always for gas dynamics.

We also assume:

(B) Dirichlet boundary conditions in $\tilde W$-coordinates:
\begin{equation}\label{inBC}
(\tilde w^I, \tilde w^{II})(0,\tx,t)=\tilde h(\tx,t):=(\tilde
h_1,\tilde h_2)(\tx,t)\end{equation} for the inflow case, and
\begin{equation}\label{outBC}
\tilde w^{II}(0,\tx,t)=\tilde h(\tx,t)\end{equation} for the outflow
case, with $x = (x_1,\tx)\in \RR^d$.
\\

This is sufficient for the main physical applications; the situation
of more general, Neumann and mixed-type boundary conditions on the
parabolic variable $v$ can be treated as discussed in
\cite{GMWZ5,GMWZ6}.

\begin{example}\label{aeroexam}
\textup{ The main example we have in mind consists of {\it laminar
solutions} $(\rho, u, e)(x_1,t)$ of the compressible Navier--Stokes
equations
\begin{equation}
\label{NSeq} \left\{ \begin{aligned}
 & \D_t \rho +  \div (\rho u) = 0
 \\
 &\D_t(\rho  u) + \div(\rho u^tu)+ \na p =
\eps \mu \Delta u + \eps(\mu+\eta) \nabla \div u
 \\
 &
 \D_t(\rho E) + \div\big( (\rho E  +p)u\big)=
\eps\kappa \Delta T +
\eps \mu \div\big( (u\cdot \nabla) u\big) \\
& \qquad \qquad \qquad \qquad \qquad \qquad + \eps(\mu+\eta)
\nabla(u\cdot \div u),
 \end{aligned}\right.
\end{equation}
$x\in \RR^d$, on a half-space $x_1>0$, where $\rho$ denotes density,
$u\in \RR^d$ velocity, $e$ specific internal energy,
$E=e+\frac{|u|^2}{2}$ specific total energy, $p=p(\rho, e)$
pressure, $T=T(\rho, e)$ temperature, $\mu>0$ and $|\eta|\le \mu$
first and second coefficients of viscosity, $\kappa>0$ the
coefficient of heat conduction, and $\eps>0$ (typically small) the
reciprocal of the Reynolds number, with no-slip {\it suction-type}
boundary conditions on the velocity,
$$
 u_{j}(0, x_2, \dots, x_d)=0, \, j\ne 1
\quad \hbox{\rm and} \quad  u_1(0,x_2, \dots, x_d)= V(x)< 0,
$$
and prescribed temperature, $ T(0,x_2, \dots, x_d)= T_{wall}(\tilde x).  $
Under the standard assumptions $p_\rho$, $T_e>0$, this can be seen
to satisfy all of the hypotheses (A1)--(A3),
(H0)--(H4), (B) in the {\it outflow case} \eqref{outBC}; 
indeed these are satisfied also under much weaker van
der Waals gas assumptions \cite{MaZ4,Z3,CHNZ,GMWZ5,GMWZ6}. In
particular, boundary-layer solutions are of noncharacteristic type,
scaling as $(\rho, u, e)= (\bar \rho, \bar u, \bar e)(x_1/\eps)$,
with layer thickness $\sim \eps$ as compared to the $\sim
\sqrt{\eps}$ thickness of the characteristic type found for an
impermeable boundary. }

\textup{ This corresponds to the situation of an airfoil with
microscopic holes through which gas is pumped from the surrounding
flow, the microscopic suction imposing a fixed normal velocity while
the macroscopic surface imposes standard temperature conditions as
in flow past a (nonporous) plate. This configuration was suggested
by Prandtl and tested experimentally by G.I. Taylor as a means to
reduce drag by stabilizing laminar flow; see \cite{S,Bra}. It was
implemented in the NASA F-16XL experimental aircraft program in the
1990's with reported 25\% reduction in drag at supersonic speeds
\cite{Bra}.\footnote{ See also NASA site
http://www.dfrc.nasa.gov/Gallery/photo/F-16XL2/index.html} Possible
mechanisms for this reduction are smaller thickness $\sim
\eps<<\sqrt{\eps}$ of noncharacteristic boundary layers as compared
to characteristic type, and greater stability, delaying the
transition from laminar to turbulent flow. In particular, stability
properties appear to be quite important for the understanding of
this phenomenon. For further discussion, including the related
issues of matched asymptotic expansion, multi-dimensional effects,
and more general boundary configurations, see \cite{GMWZ5}. }
\end{example}

\begin{example}\label{aeroexamin}
\textup{Alternatively, we may consider the compressible Navier--Stokes 
equations \eqref{NSeq} with {\it blowing-type} boundary conditions
$$
 u_{j}(0, x_2, \dots, x_d)=0, \, j\ne 1
\quad \hbox{\rm and} \quad  u_1(0,x_2, \dots, x_d)= V(x)> 0,
$$
and prescribed temperature and pressure
$$
\begin{aligned}
 T(0,x_2, \dots, x_d)&= T_{wall}(\tilde x),
\quad
 p(0,x_2, \dots, x_d)= p_{wall}(\tilde x)\\   
\end{aligned}
$$
(equivalently, prescribed temperature and density).
Under the standard assumptions $p_\rho$, $T_e>0$
on the equation of state (alternatively,
van der Waals gas assumptions), this can be seen
to satisfy hypotheses (A1)--(A3),
(H0)--(H4), (B) in the {\it inflow case} \eqref{inBC}. 
}
\end{example}

\begin{lemma}[\cite{MaZ3,Z3,GMWZ5,NZ}]\label{lem-profile-decay}
Given (A1)-(A3) and (H0)-(H2), a standing wave solution
\eqref{profile} of \eqref{sys}, (B) satisfies
\begin{equation}\label{layerdecay}
\Big|(d/dx_1)^k(\bU - U_+)\Big|\le C e^{-\theta x_1},
\quad 0\le k\le s+1,
\end{equation}
as $x_1\to +\infty$, $s$ as in (H0). Moreover, a solution, if it
exists, is in the inflow or strictly parabolic case unique; in the
outflow case it is locally unique.
\end{lemma}

\begin{proof}
See Lemma 1.3, \cite{NZ}.
\end{proof}

\subsection{The Evans condition and strong spectral stability}
The linearized equations of \eqref{sys}, (B) about $\bar U$ are
\begin{equation}\label{lind}
 U_t = LU :=
\sum_{j,k}( B^{jk} U_{x_k})_{x_j} - \sum_{j}( A^{j} U)_{x_j}
\end{equation}
with initial data $U(0)=U_0$ and boundary conditions in (linearized)
$\tilde W$-coordinates of
$$
W(0,\tx,t):=( w^I, w^{II})^T(0,\tx,t)=h
$$
for the inflow case, and
$$
 w^{II}(0,\tx,t)=h
$$
for the outflow case, with $x = (x_1,\tx)\in \RR^d$, where
$W:=(\partial \tilde W/\partial U)(\bar U)U$.

A necessary condition for linearized stability is weak spectral
stability, defined as nonexistence of unstable spectra $\Re \lambda
>0$ of the linearized operator $L$ about the wave. As described in
Section \ref{evanssubsub}, this is equivalent to nonvanishing for
all $\tilde \xi\in \RR^{d-1}$, $\Re \lambda>0$ of the {\it Evans
function}
$$
D_L(\tilde \xi, \lambda)
$$
(defined in \eqref{eq:Evans}), a Wronskian associated with the
Fourier-transformed eigenvalue ODE.

\begin{definition}\label{strongspectral}
\textup{ We define {\it strong spectral stability} as {\it uniform
Evans stability}:
\begin{equation}\tag{D}
|D_L(\tilde \xi, \lambda)|\ge \theta(C)>0
\end{equation}
for $(\tilde \xi, \lambda)$ on bounded subsets $C\subset \{\tilde
\xi\in \RR^{d-1}, \, \Re \lambda \ge 0\}\setminus\{0\}$. }
\end{definition}

For the class of equations we consider, this is equivalent to the
uniform Evans condition of \cite{GMWZ5,GMWZ6}, which includes an
additional high-frequency condition that for these equations is
always satisfied
(see Proposition 3.8, \cite{GMWZ5}).
A fundamental result proved in \cite{GMWZ5} is that small-amplitude
noncharacteristic boundary-layers are always strongly spectrally
stable.\footnote{
The result of \cite{GMWZ5} 
applies also to more general types of boundary conditions
and in some situations to systems with variable multiplicity
characteristics, including, in some parameter ranges, MHD.}

\begin{prop}[\cite{GMWZ5}]\label{specstab}
Assuming (A1)-(A3),
(H0)-(H3), (B)
for some fixed endstate (or compact set of endstates) $U_+$,
boundary layers with amplitude
$$
\|\bar U-U_+\|_{L^\infty[0,+\infty]}
$$
sufficiently small satisfy the strong spectral stability condition
(D).
\end{prop}

As demonstrated in \cite{SZ}, stability of large-amplitude boundary
layers may fail for the class of equations considered here, even in
a single space dimension, so there is no such general theorem in the
large-amplitude case. Stability of large-amplitude boundary-layers
may be checked efficiently by numerical Evans computations as in
\cite{BDG,Br1,Br2,BrZ,HuZ,BHRZ,HLZ,CHNZ,HLyZ1,HLyZ2}.

\subsection{Main results}
Our main results are as follows.

\begin{theorem}[Linearized stability]\label{theo-lin}
Assuming (A1)-(A3), (H0)-(H4), (B), and strong spectral
stability (D), we obtain asymptotic 
$L^1\cap H^{[(d-1)/2]+5} \rightarrow
L^p$ stability of \eqref{lind}
in dimension $d\ge 2$, and any $2\le p\le \infty$,
with rate of decay
\begin{equation}\begin{aligned} |U(t)|_{L^2}&\le C
(1+t)^{-\frac{d-1}{4}} (|U_0|_{L^1\cap H^{3}} +E_0),\\ |U(t)|_{L^p}&\le C
(1+t)^{-\frac{d}{2}(1-1/p)+1/2p} 
(|U_0|_{L^1\cap H^{[(d-1)/2]+5}}+E_0),
\end{aligned}\end{equation}
provided that the initial perturbations $U_0$ are in $L^1\cap H^3$ for $p=2$, or 
in $L^1 \cap H^{[(d-1)/2]+5}$ for $p>2$, and boundary perturbations
$h$ satisfy
\begin{equation}\label{decayh}\begin{aligned}|h(t)|_{L^2_{\tx}}&\le E_0(1+t)^{-(d+1)/4},
\\ |h(t)|_{L^\infty_{\tx}}&\le E_0(1+t)^{-d/2}
\\ |\cD_h(t)|_{L^1_{\tx} \cap H^{[(d-1)/2]+5}_{\tx}}
 &\le E_0(1+t)^{-d/2 - \epsilon},\end{aligned}\end{equation} where $\cD_h(t):=|h_t|+|h_{\tx}|+|h_{\tx\tx}|$, $E_0$ is some positive constant, and $\epsilon >0$ is arbitrary small for the case $d=2$
and $\epsilon =0$ for $d\ge 3$.
\end{theorem}

\begin{theorem}[Nonlinear stability]\label{theo-nonlin}
Assuming (A1)-(A3), (H0)-(H4), (B), and strong spectral
stability (D), we obtain asymptotic $L^1\cap H^s \rightarrow L^p
\cap H^s$ stability of $\bar U$ as a solution of \eqref{sys}
in dimension $d\ge 2$, for $s\ge s(d)$ as
defined in (H0), and any $2\le p\le \infty$, with rate of decay
\begin{equation}\begin{aligned} &|\tilde U(t)-\bU|_{L^p}\le C
(1+t)^{-\frac{d}{2}(1-1/p)+1/2p} (|U_0|_{L^1\cap H^s}+E_0)\\&|\tilde
U(t)-\bU|_{H^s}\le C (1+t)^{-\frac{d-1}4} (|U_0|_{L^1\cap H^s}+E_0),
\end{aligned}\end{equation}
provided that the initial perturbations $U_0:=\tilde U_0 - \bU$ are
sufficiently small in $L^1 \cap H^s$ and boundary perturbations
$h(t):=\tilde h(t) - W(\bU_0)$ satisfy \eqref{decayh} and
\begin{equation} \CalB_h(t) \le E_0
(1+t)^{-\frac{d-1}{4}},\end{equation} with sufficiently small $E_0$,
where the boundary measure $\CalB_h$ is defined as
\begin{equation}\label{Bdry-out}
\begin{aligned}\CalB_h(t):=
 |h|_{H^s(\tx)}+\sum_{i=0}^{[(s+1)/2]}|\partial_t^ih|_{L^2(\tx)}
\end{aligned}\end{equation}
for the outflow case, and similarly
\begin{equation}\label{Bdry-in} \CalB_h(t):=
|h|_{H^s(\tx)}+ \sum_{i=0}^{[(s+1)/2]}|\partial_t^ih_2|_{L^2(\tx)}+
\sum_{i=0}^s|\partial_t^ih_1|_{L^2(\tx)}
\end{equation}
for the inflow case.
\end{theorem}

Combining
Theorem \ref{theo-nonlin} and Proposition \ref{specstab}, we obtain the
following small-amplitude stability result, applying in particular
to the motivating situation of Example \ref{aeroexam}.

\begin{cor}\label{smallamp}
Assuming (A1)-(A3), (H0)-(H4), (B) for some fixed endstate (or
compact set of endstates) $U_+$, boundary layers with amplitude
$$
\|\bar U-U_+\|_{L^\infty[0,+\infty]}
$$
sufficiently small are linearly and nonlinearly stable in the sense
of Theorems \ref{theo-lin} and \ref{theo-nonlin}.
\end{cor}

\begin{remark}\label{rate}
\textup{ The obtained rate of decay in $L^2$ may be recognized as
that of a $(d-1)$-dimensional heat kernel, and the obtained rate of
decay in $L^\infty$ as that of a $d$-dimensional heat kernel. We
believe that the sharp rate of decay in $L^2$ is rather that of a
$d$-dimensional heat kernel and the sharp rate of decay in
$L^\infty$ dependent on the characteristic structure of the
associated inviscid equations, as in the constant-coefficient case
\cite{HoZ1,HoZ2}. }
\end{remark}

\begin{remark}\label{nec}
\textup{
In one dimension, strong spectral stability is necessary for
linearized asymptotic stability;
see Theorem 1.6, \cite{NZ}. However, in multi-dimensions, it appears
likely that, as in the shock case \cite{Z3}, there are intermediate
possibilities between strong and weak spectral stability for which
linearized stability might hold with degraded rates of decay.
In any case, the gap between the necessary weak spectral
and the sufficient strong spectral stability conditions concerns
only pure imaginary spectra $\Re \lambda=0$ on the boundary between
strictly stable and unstable half-planes,
so this should not interfere with investigation of
physical stability regions.
%
}
\end{remark}

\subsection{Discussion and open problems}

Asymptotic stability, without rates of decay, has been shown for
small amplitude noncharacteristic ``normal'' boundary layers of the isentropic
compressible Navier--Stokes equations 
with outflow boundary conditions and vanishing transverse velocity
in \cite{KK}, using energy estimates. 
Corollary \ref{smallamp} recovers this existing result
and extends it to the general arbitrary transverse velocity, 
outflow or inflow, and 
isentropic or nonisentropic (full compressible
Navier--Stokes) case, in addition giving asymptotic rates of decay.
Moreover, we treat perturbations of boundary as well as initial
data, as previous time-asymptotic investigations (with the exception
of direct predecessors \cite{YZ,NZ}) do not.
As discussed in Appendix \ref{phys}, 
the type of boundary layer relevant to the drag-reduction strategy
discussed in Examples \ref{aeroexam}--\ref{aeroexamin} is a
noncharacteristic ``transverse'' type with constant normal velocity,
complementary to the normal type considered in \cite{KK}.

The large-amplitude asymptotic stability result of Theorem
\ref{theo-nonlin} extends to multi dimensions corresponding
one-dimensional results of \cite{YZ,NZ}, reducing the problem of
stability to verification of a numerically checkable Evans
condition. See also the related, but technically rather different,
work on the small viscosity limit in \cite{MZ,GMWZ5,GMWZ6}.
By a combination of numerical Evans function computations and
asymptotic ODE estimates, spectral stability has been checked for
{\it arbitrary amplitude} noncharacteristic boundary layers of the
one-dimensional isentropic compressible Navier--Stokes equations in
\cite{CHNZ}. Extensions to the nonisentropic and multi-dimensional
case should be possible by the methods used in \cite{HLyZ1} and
\cite{HLyZ2} respectively to treat the related shock stability
problem. 

This (investigation of large-amplitude spectral stability) 
would be a very interesting direction for further investigation.
In particular, note that 
it is large-amplitude stability that is relevant to drag-reduction
at flight speeds, since the transverse relative velocity (i.e., velocity
parallel to the airfoil) is zero at the wing surface and flight speed
outside a thin boundary layer, so that variation across the
boundary layer is substantial.
We discuss this problem further in Appendix \ref{phys}
for the model isentropic case.

Our method of analysis follows the basic approach introduced in
\cite{Z2,Z3,Z4} for the study of multi-dimensional shock stability
and we are able to make use of much of that analysis without
modification. However, there are some new difficulties to be
overcome in the boundary-layer case.

The main new difficulty is that the boundary-layer case is analogous
to the {\it undercompressive shock} case rather than the more
favorable {\it Lax shock} case emphasized in \cite{Z3}, in that
$G_{y_1} \not \sim  t^{-1/2}G$ as in the Lax shock case but rather
$G_{y_1} \sim  (e^{-\theta |y_1|}+t^{-1/2})G$, $\theta>0$, as in the
undercompressive case. This is a significant difficulty; indeed, for
this reason, the undercompressive shock analysis was carried out in
\cite{Z3} only in nonphysical dimensions $d\ge 4$. On the other
hand, there is no translational invariance in the boundary layer
problem, so no zero-eigenvalue and no pole of the resolvent kernel
at the origin for the one-dimensional operator, and in this sense
$G$ is somewhat better in the boundary layer than in the shock case.

Thus, the difficulty of the present problem is roughly intermediate
to that of the  Lax and undercompressive shock cases. Though the
undercompressive shock case is still open in multi-dimensions for
$d\le 3$, the slight advantage afforded by lack of pole terms allows
us to close the argument in the boundary-layer case. Specifically,
thanks to the absence of pole terms, we are able to get a slightly
improved rate of decay in $L^\infty(x_1)$ norms, though our
$L^2(x_1)$ estimates remain the same as in the shock case. By
keeping track of these improved sup norm bounds throughout the
proof, we are able to close the argument without using detailed
pointwise bounds as in the one-dimensional analyses of \cite{HZ,RZ}.

Other difficulties include the appearance of boundary terms in
integrations by parts, which makes the auxiliary energy estimates by
which we control high-frequency effects considerably
more difficult in the boundary-layer than in the shock-layer case,
and
the treatment of boundary perturbations. In terms of the homogeneous
Green function $G$, boundary perturbations lead by a standard
duality argument to contributions consisting of integrals on the
boundary of perturbations against various derivatives of $G$, and
these are a bit too singular as time goes to zero to be absolutely
integrable. Following the strategy introduced in \cite{YZ,NZ}, we
instead use duality to convert these to less singular integrals over
the whole space, that {\it are} absolutely integrable in time.
However, we make a key improvement here over the treatment in
\cite{YZ,NZ}, integrating against an exponentially decaying test
function to obtain terms of exactly the same form already treated
for the homogeneous problem.  This is necessary for us in the
multi-dimensional case, for which we have insufficient information
about individual parts of the solution operator to estimate them
separately as in \cite{YZ,NZ}, but makes things much more
transparent also in the one-dimensional case.

Among physical systems, our hypotheses appear to apply to and
essentially only to the case of compressible Navier--Stokes
equations with inflow or outflow boundary conditions.  However, the
method of analysis should apply, with suitable modifications, to
more general situations such as MHD; see for example the recent
results on the related small-viscosity problem in
\cite{GMWZ5,GMWZ6}. The extension to MHD is a very interesting open
problem.

Finally, as pointed out in Remark \ref{nec}, the strong spectral
stability condition does not appear to be necessary for asymptotic
stability.  It would be 
interesting to develop a 
refined stability condition similarly as was done in \cite{SZ,Z2,Z3,Z4} 
for the shock case.

\section{Resolvent kernel: construction and low-frequency bounds} In
this section, we briefly recal the construction of resolvent kernel
and then establish the pointwise low-frequency bounds on
$G_{\txi,\lambda}$, by appropriately modifying the proof in
\cite{Z3} in the boundary layer context \cite{YZ,NZ}.

\subsection{Construction}
We construct a representation for the family of elliptic Green
distributions $G_{\txi,\lambda}(x_1,y_1)$, \begin{equation}
G_{\txi,\lambda}(\cdot,y_1):=(L_\txi
-\lambda)^{-1}\delta_{y_1}(\cdot),\end{equation} associated with the
ordinary differential operators $(L_\txi -\lambda)$, i.e. the
resolvent kernel of the Fourier transform $L_\txi$ of the linearized
operator $L$ of \eqref{lind}. To do so, we study the homogeneous
eigenvalue equation $(L_\txi - \lambda)U =0$, or
\begin{equation}\label{eg-eqs}\begin{aligned}\overbrace{(B^{11}U')'-(A^1U)'}^{L_0U} - &i\sum_{j\not=1}A^j\xi_jU +
i\sum_{j\not=1}B^{j1}\xi_jU' \\&+i\sum_{k\not=1}(B^{1k}\xi_kU)' -
\sum_{j,k\not=1}B^{jk}\xi_j\xi_k U - \lambda U =0,
\end{aligned}\end{equation}
with boundary conditions (translated from those in $W$-coordinates)
\begin{equation}
\begin{pmatrix}A^1_{11}-A^1_{12} (b^{11}_2)^{-1}b^{11}_1&0\\b^{11}_1& b^{11}_2\end{pmatrix}
U(0)\equiv
\begin{pmatrix}*\\0\end{pmatrix}
\end{equation}
where $*=0$ for the inflow case and is arbitrary for the outflow
case.

Define $$\Lambda^\txi:= \bigcap_{j=1}^n\Lambda_j^+(\txi)$$ where
$\Lambda_j^+(\txi)$ denote the open sets bounded on the left by the
algebraic curves $\lambda_j^+(\xi_1,\txi)$ determined by the
eigenvalues of the symbols $-\xi^2 B_+ - i\xi A_+$ of the limiting
constant-coefficient operators
$$L_{\txi+}w:=B_+ w'' - A_+w'$$ as $\xi_1$ is varied along the real axis, with $\txi$ held fixed. The curves $\lambda_j^+(\cdot,\txi)$ comprise
the essential spectrum of operators $L_{\txi+}$. Let $\Lambda$
denote the set of $(\tilde \xi, \lambda)$ such that $\lambda\in
\Lambda^{\tilde \xi}$.

For $(\tilde \xi, \lambda)\in \Lambda^{\tilde \xi}$, introduce
locally analytically chosen (in $\tilde \xi$, $\lambda$) matrices
\begin{equation} \label{phi0+} \Phi^+= (\phi^+_1,
\cdots ,\phi^+_{k})
, \quad \Phi^0= ( \phi^0_{k+1}, \cdots , \phi^0_{n+r} ),
\end{equation} and \begin{equation} \label{phi} \Phi=
(\Phi^+,\Phi^0),
\end{equation}
whose columns span the subspaces of solutions of \eqref{eg-eqs}
that, respectively, decay at $x=+\infty$ and satisfy the prescribed
boundary conditions at $x=0$,
and locally analytically chosen matrices
\begin{equation}\label{psi0+}
\Psi^0= ( \psi^0_1, \cdots , \psi^0_{k}), \quad \Psi^+= (
\psi^+_{k+1} , \cdots , \psi^+_{n+r})
\end{equation}
and \begin{equation}\label{psi} \Psi= ( \Psi^0, \Psi^+).
\end{equation}
whose columns span complementary subspaces. The existence of such
matrices is guaranteed by the general Evans function framework of
\cite{AGJ,GZ,MaZ3}; see in particular \cite{Z3,NZ}. That dimensions
sum to $n+r$ follows by a general result of \cite{GMWZ5}; see also
\cite{SZ}.

\subsubsection{The Evans function}\label{evanssubsub}
Following \cite{AGJ,GZ,SZ}, we define on $\Lambda$ the {\it Evans
function}
\begin{equation}\label{eq:Evans}
D_L(\tilde \xi, \lambda):=\det(\Phi^0,\Phi^+)_{|x=0}.
\end{equation}
Evidently, eigenfunctions decaying at
$+\infty$ and satisfying the prescribed boundary conditions at $x_1=0$
 occur precisely when the subspaces $\Span \Phi^0$
and $\Span \Phi^+$ intersect, i.e., at zeros of the Evans function
$$
D_L(\tilde \xi, \lambda)=0.
$$

The Evans function as constructed here is locally analytic in
$(\tilde \xi, \lambda)$, which is all that we need for our analysis;
we prescribe different versions of the Evans function as needed on
different neighborhoods of $\Lambda$. Note that $\Lambda$ includes
all of $\{\tilde \xi\in \RR^{d-1}$, $ \Re \lambda \ge
0\}\setminus\{0\}$, so that Definition \ref{strongspectral} is
well-defined and equivalent to simple nonvanishing, away from the
origin $(\tilde \xi,\lambda)=(0,0)$. To make sense of this
definition near the origin, we must insist that the matrices
$\Phi^j$ in \eqref{eq:Evans} remain {\it uniformly bounded}, a
condition that can always be achieved by limiting the neighborhood
of definition.

For the class of equations we consider, the Evans function may in
fact be extended continuously along rays through the origin
\cite{R2,MZ,GMWZ5,GMWZ6}.

\subsubsection{Basic representation formulae}\label{representationforms}
Define the solution operator from $y_1$ to $x_1$ of ODE
$(L_\txi-\lambda)U =0$, denoted by $\CalF^{y_1\rightarrow x_1}$, as
$$\CalF^{y_1\rightarrow x_1} = \Phi (x_1,\lambda)\Phi^{-1}(y_1,\lambda)$$
and the projections $\Pi_{y_1}^0,\Pi_{y_1}^+$ on the stable
manifolds at $0,+\infty$ as
$$\Pi_{y_1}^+ =\begin{pmatrix}\Phi^+(y_1) & 0
\end{pmatrix}\Phi^{-1}(y_1), \quad \Pi_{y_1}^0 =\begin{pmatrix} 0& \Phi^0(y_1)
\end{pmatrix}\Phi^{-1}(y_1).$$

%

We define also the dual subspaces of solutions of $(L^*_\txi -
\lambda ^*)\tilde W =0$. We denote growing solutions
\begin{equation}
\tilde{\Phi}^0= (\tilde{\phi}^0_1 , \cdots ,
\tilde{\phi}^0_{k}),\quad \tilde{\Phi}^+= (\tilde{\phi}^+_{k+1},
\cdots , \tilde{\phi}^+_{n+r}),
\end{equation} $\tilde{\Phi}:=(\tilde{\Phi}^0,\tilde{\Phi}^+)$ and decaying solutions
\begin{equation}
\tilde{\Psi}^0=(\tilde{\psi}^0_1,\cdots , \tilde{\psi}^+_{k}) ,\quad
\tilde{\Psi}^+= (\tilde{\psi}^+_{k+1} , \cdots ,
\tilde{\psi}^+_{n+r}),
\end{equation} and $\tilde{\Psi}:=(\tilde{\Psi}^0,\tilde{\Psi}^+)$,
satisfying the relations
\begin{equation}\label{duality}\begin{pmatrix}\tilde \Psi &\tilde
\Phi\end{pmatrix}_{0,+}^* \bar \CalS^\txi \begin{pmatrix}\Psi
&\Phi\end{pmatrix}_{0,+} \equiv I,\end{equation} where
\begin{equation}\bar \CalS^\txi = \begin{pmatrix} -A^{1} + iB^{1\txi} + iB^{\txi 1} & \begin{pmatrix}0\\I_r\end{pmatrix} \\
\begin{pmatrix}-(b_2^{11})^{-1}b_I^{11} & -I_r\end{pmatrix}&0\end{pmatrix}.\end{equation}

With these preparations, the construction of the Resolvent kernel
goes exactly as in the construction performed in \cite{ZH,MaZ3,Z3}
on the whole line and \cite{YZ,NZ} on the half line, yielding the
following basic representation formulae; for a proof, see
\cite{MaZ3,NZ}.

\begin{prop}\label{prop-G-rep}
We have the following representation
\begin{equation}\label{repG}
G_{\txi,\lambda} (x_1,y_1)=
\begin{cases}
   (I_n, 0)\CalF^{y_1\rightarrow x_1}\Pi_{y_1}^+ ({\bar
   S}^\txi)^{-1}(y_1)(I_n,0)^{tr},\quad & for \quad x_1>y_1,\\
-(I_n, 0)\CalF^{y_1\rightarrow x_1}\Pi_{y_1}^0 ({\bar
   S}^\txi)^{-1}(y_1)(I_n,0)^{tr},\quad &for \quad x_1<y_1.
\end{cases}
\end{equation}
\end{prop}

\begin{prop}\label{prop-G-rep1} The resolvent kernel may alternatively be expressed as
$$
  G_{\txi,\lambda}(x_1,y_1)=
  \begin{cases}
   (I_n,0)\Phi^+(x_1;\lambda)M^+(\lambda)\tilde\Psi^{0*}(y_1;\lambda)(I_n,0)^{tr}\ &x_1>y_1,\\
   -(I_n,0)\Phi^0(x_1;\lambda)M^0(\lambda)\tilde\Psi^{+*}(y_1;\lambda)(I_n,0 )^{tr}\ &x_1<y_1,
                 \end{cases}
$$
where
\begin{equation}\label{Mexp}
  M(\lambda):=\text{\rm diag}(M^+(\lambda),M^0(\lambda))=
  \Phi^{-1}(z;\lambda)(\bar \CalS^{\txi})^{-1}(z)\tilde\Psi^{-1*}(z;\lambda).
\end{equation}
\end{prop}

\subsubsection{Scattering decomposition}\label{scattering}
From Propositions \ref{prop-G-rep} and \ref{prop-G-rep1}, we obtain
the following scattering decomposition, generalizing the Fourier
transform representation in the constant-coefficient case, from
which we will obtain pointwise bounds in the low-frequency regime.

\begin{cor}\label{cor-Grep} On $\Lambda^\txi \cap \rho (L_\txi)$,
\begin{equation}\label{G-rep3} G_{\txi,\lambda}(x_1,y_1)=
\sum_{j,k}d_{jk}^+\phi_j^+(x_1;\lambda)\tilde
\psi_k^+(y_1;\lambda)^* + \sum_k \phi^+_k(x_1;\lambda)\tilde
\phi_k^+(y_1;\lambda)^*\end{equation}for $0\le y_1\le x_1$, and
\begin{equation}\label{G-rep4} G_{\txi,\lambda}(x_1,y_1)=
\sum_{j,k}d_{jk}^0\phi_j^+(x_1;\lambda)\tilde
\psi_k^+(y_1;\lambda)^* - \sum_k \psi_k^+(x_1;\lambda)\tilde
\psi_k^+(y_1;\lambda)^*\end{equation}for $0\le x_1\le y_1$, where
\begin{equation}\label{djk}\begin{aligned}d_{jk}^{0,+}(\lambda)&= (I,0)\begin{pmatrix}\Phi^+ & \Phi^0\end{pmatrix}^{-1}\Psi^+.\end{aligned}\end{equation}
\end{cor}

\begin{proof} For $0\le x_1\le y_1$, we obtain the preliminary
representation $$G_{\txi,\lambda}(x_1,y_1)=
\sum_{j,k}d_{jk}^0(\lambda)\phi_j^+(x_1;\lambda)\tilde
\psi_k^+(y_1;\lambda)^* + \sum_{jk}e_{jk}^0
\psi_j^+(x_1;\lambda)\tilde \psi_k^+(y_1;\lambda)^*$$ from which,
together with duality \eqref{duality}, representation \eqref{repG},
and the fact that $\Pi_0 = I - \Pi_+$, we have
\begin{equation}\begin{aligned} \begin{pmatrix}d^0 \\e^0\end{pmatrix} &= -\begin{pmatrix}\tilde \Phi^+ & \tilde \Psi^+\end{pmatrix}^* A \Pi_0 \Psi^+ \\&=-\begin{pmatrix}\Phi^+ & \Psi^+\end{pmatrix}^{-1}\Big[I- \begin{pmatrix}\Phi^+ & 0\end{pmatrix} \begin{pmatrix}\Phi^+ & \Phi^0\end{pmatrix} ^{-1}
\Big]\Psi^+ \\&= \begin{pmatrix}0\\-I_k\end{pmatrix} +
\begin{pmatrix}I_{n-k} & 0 \\0&0\end{pmatrix}\begin{pmatrix}\Phi^+ &
\Phi^0\end{pmatrix} ^{-1} \Psi^+ .
\end{aligned}\end{equation}

Similarly, for $0\le y_1\le x_1$, we obtain the preliminary
representation $$G_{\txi,\lambda}(x_1,y_1)=
\sum_{j,k}d_{jk}^+(\lambda)\phi_j^+(x_1;\lambda)\tilde
\psi_k^+(y_1;\lambda)^* + \sum_{jk}e_{jk}^+
\phi_j^+(x_1;\lambda)\tilde \phi_k^+(y_1;\lambda)^*$$ from which,
together with duality \eqref{duality} and representation
\eqref{repG}, we have
\begin{equation}\begin{aligned} \begin{pmatrix}d^+ \\e^+\end{pmatrix} &= \tilde \Phi^{+*} A \Pi_+ \begin{pmatrix}\Psi^+ & \Phi^+\end{pmatrix} \\&=(\Phi^+)^{-1} \begin{pmatrix}\Phi^+ & 0\end{pmatrix} \begin{pmatrix}\Phi^+ & \Phi^0\end{pmatrix} ^{-1} \begin{pmatrix}\Psi^+ & \Phi^+\end{pmatrix}\\&=\begin{pmatrix}I & 0\end{pmatrix} \begin{pmatrix}\Phi^+ & \Phi^0\end{pmatrix} ^{-1} \begin{pmatrix}\Psi^+ & \Phi^+\end{pmatrix}\\&=\begin{pmatrix}I_{n-k} & 0\\0&0\end{pmatrix} \begin{pmatrix}\Phi^+ & \Phi^0\end{pmatrix} ^{-1}\Psi^+ + \begin{pmatrix}0&0\\I_k&0\end{pmatrix} \begin{pmatrix}0 & I_k\\0&0\end{pmatrix}.
\end{aligned}\end{equation}
\end{proof}

\begin{remark}
\textup{ In the constant-coefficient case, with a choice of common
bases $\Psi^{0,+} = \Phi^{+,0}$ at $0,+\infty$, the above
representation reduces to the simple formula
\begin{equation}G_{\txi,\lambda}(x_1,y_1) =
\begin{cases}
   \sum_{j=k+1}^N\phi_j^+(x_1;\lambda)\tilde \phi_j^{+*}(y_1;\lambda) &x_1>y_1,\\
   -\sum_{j=1}^k\psi_j^+(x_1;\lambda)\tilde \psi_j^{+*}(y_1;\lambda) &x_1<y_1.
                 \end{cases}
\end{equation}
}
\end{remark}

\subsection{Pointwise low-frequency bounds}
We obtain pointwise low-frequency bounds on the resolvent kernel
$G_{\txi,\lambda}(x_1,y_1)$ by appealing to the detailed analysis of
\cite{Z2,Z3,GMWZ1} in the viscous shock case. Restrict attention to
the surface
\begin{equation}\Gamma^{\tilde \xi}:=\{\lambda~:~\R\lambda =
-\theta_1(|\tilde \xi|^2 + |\I\lambda|^2)\},\end{equation} for
$\theta_1>0$ sufficiently small.

\begin{proposition}[\cite{Z3}]\label{prop-Gbounds}
Under the hypotheses of Theorem \ref{theo-nonlin}, for $\lambda \in
\Gamma^\txi$ and $\rho:=|(\txi,\lambda)|$, $\theta_1>0$, and
$\theta>0$ sufficiently small, there hold:

\begin{equation}\label{G1}
|G_{\tilde \xi,\lambda}(x_1,y_1)| \le C\gamma_2
e^{-\theta\rho^2|x_1-y_1|}.\end{equation} and
\begin{equation}\label{G2}
|\partial_{y_1}^\beta G_{\tilde \xi,\lambda}(x_1,y_1)| \le C\gamma_2
(\rho^{\beta}+\beta e^{-\theta
y_1})e^{-\theta\rho^2|x_1-y_1|}\end{equation} where
\begin{equation}\label{gamma}
\gamma_2:=1+\sum_{j}\Big[\rho^{-1}|\I\lambda -
\eta_j^+(\txi)|+\rho\Big]^{1/s_j-1},\end{equation}and
$s_j,\eta_j^+(\txi)$ are as defined in (H4).
\end{proposition}

\begin{proof}
This follows by a simplified version of the analysis of \cite{Z3},
Section 5 in the viscous shock case, replacing $\Phi^-$, $\Psi^-$
with $\Phi^0$, $\Psi^0$, omitting the refined derivative bounds of
Lemmas 5.23 and 5.27 describing special properties of the Lax and
overcompressive shock case (not relevant here), and setting
$\ell=0$, or $\tilde \gamma\equiv 1$ in definition (5.128). Here,
$\ell$ is the multiplicity to which the Evans function vanishes at
the origin, $(\tilde \xi, \lambda)=(0,0)$, evidently zero under
assumption (D). The key modes $\Phi^+$, $\Psi^+$ at plus spatial
infinity are the same for the boundary-layer as for the shock case.

This leads to the pointwise bounds (5.37)--(5.38) given in
Proposition 5.10 of \cite{Z3} in case
$\alpha=1$, $\gamma_1\equiv 1$ corresponding to the uniformly stable
undercompressive shock case, but without the first $O(\rho^{-1})$,
or ``pole'', terms appearing on the righthand side, which derive
from cases $\tilde \gamma\sim \rho^{-1}$ not arising here. But,
these are exactly the claimed bounds \eqref{G1}--\eqref{gamma}.

We omit the (substantial) details of this computation, referring the
reader to \cite{Z3}. However, the basic idea is, starting with the
scattering decomposition of Corollary \ref{scattering}, to note,
first, that the normal modes $\Phi^j$, $\Psi^j$, $\tilde \Phi^j$,
$\tilde\Psi^j$ can be approximated up to an exponentially trivial
coordinate change by solutions of the constant-coefficient limiting
system at $x\to +\infty$ (the conjugation lemma of \cite{MZ}) and,
second, that the coefficients $M_{jk}$, $d_{jk}$ may be
well-estimated through formulae \eqref{Mexp} and \eqref{djk} using
Kramer's rule and the assumed lower bound on th Evans function $|D|$
appearing in the denominator. This is relatively straightforward away from the
branch points $\Im \lambda= \eta_j(\tilde \xi)$ or ``glancing set''
of hyperbolic theory; the treatment near these points involves some
delicate matrix perturbation theory applied to the limiting
constant-coefficient system at $x\to +\infty$ followed by careful
bookkeeping in the application of Kramer's rule.
\end{proof}

\section{Linearized estimates} We next establish estimates on
the linearized inhomogeneous problem
\begin{equation} \label{inhom}
U_t - LU = f\end{equation} with
initial data $U(0)=U_0$ and Dirichlet boundary conditions as usual
in $\tilde W$-coordinates:
\begin{equation}\label{zerobdry-in}
W(0,\tx,t):=( w^I,  w^{II})^T(0,\tx,t)=h\end{equation} for the
inflow case, and
\begin{equation}\label{zerobdry-out}
w^{II}(0,\tx,t)=h\end{equation} for the outflow case,
 with $x =(x_1,\tx)\in \RR^d$.

\subsection{Resolvent bounds} Our first step is
to estimate solutions of the resolvent equation
with {homogeneous} boundary data $\hat h\equiv 0$.

\begin{proposition}[High-frequency bounds]\label{prop-resHF} Given (A1)-(A2),
(H0)-(H2), and homogeneous boundary conditions (B), for some $R,C$
sufficiently large and $\theta>0$ sufficiently small,
\begin{equation}\label{oldres-est}
|(L_\txi - \lambda)^{-1} \hat f|_{\hat H^1(x_1)} \le C |\hat f|_{\hat
H^1(x_1)},
\end{equation}
and
\begin{equation}\label{res-est} |(L_\txi -
\lambda)^{-1}\hat f|_{L^2(x_1)} \le \frac{C}{|\lambda|^{1/2}} |\hat
f|_{\hat H^1(x_1)},
\end{equation}
for all $|(\txi,\lambda)|\ge R$ and $\R\lambda \ge
-\theta$, where $\hat f$ is the Fourier transform of $f$ in variable
$\tx$ and $|\hat f|_{\hat H^1(x_1)} :=
|(1+|\partial_{x_1}|+|\txi|)\hat f|_{L^2(x_1)}$.
\end{proposition}

\begin{proof} First observe that a Laplace-Fourier transformed version with respect to variables $(\lambda,\tx)$ of the
nonlinear energy estimate in Section \ref{sec-EE} with $s = 1$,
carried out on the linearized equations written in $W$-coordinates,
yields
\begin{equation} \begin{aligned} (\R\lambda &+
\theta_1)|(1+|\txi|+|\partial_{x_1}|)W|^2 \le C\Big( |W|^2 +
(1+|\txi|^2)|W||\hat f| + |\partial_{x_1}W||\partial_{x_1} \hat
f|\Big)\end{aligned}\end{equation} for some $C$ big and $\theta_1>0$
sufficiently small, where $|.|$ denotes $|.|_{L^2(x_1)}$. Applying
Young's inequality, we obtain
\begin{equation}\label{Re-est} \begin{aligned} (\R\lambda &+
\theta_1)|(1+|\txi|+|\partial_{x_1}|)W|^2\le C|W|^2 +
C|(1+|\txi|+|\partial_{x_1}|)\hat f|^2.\end{aligned}\end{equation}

On the other hand, taking the imaginary part of the $L^2$ inner
product of $U$ against $\lambda U = f+LU$, we have also the standard
estimate
\begin{equation} |\I\lambda||U|_{L^2}^2\le C|U|_{H^1}^2 + C|f|_{L^2}^2,\end{equation}
and thus, taking the Fourier transform in $\tx$, we obtain
\begin{equation} \label{Im-est}|\I\lambda||W|^2 \le C|\hat f|^2 +
C|(1+|\txi|+|\partial_{x_1}|)W|^2.\end{equation}

Therefore, taking $\theta=\theta_1/2$, we obtain from \eqref{Re-est}
and \eqref{Im-est}
\begin{equation}\begin{aligned} |(1+|\lambda|^{1/2}+|\txi|+|\partial_{x_1}|)W|^2\le C|W|^2 +
C|(1+|\txi|+|\partial_{x_1}|)\hat f|^2,\end{aligned}\end{equation}
for any $\R\lambda \ge -\theta$. Now take $R$ sufficiently large
such that $|W|^2$ on the right hand side of the above can be
absorbed into the left hand side,
and thus, for all $|(\txi,\lambda)|\ge R$ and $\R\lambda \ge
-\theta$,
\begin{equation}\begin{aligned}
|(1+|\lambda|^{1/2}+|\txi|+|\partial_{x_1}|)W|^2\le
C|(1+|\txi|+|\partial_{x_1}|)\hat f|^2,\end{aligned}\end{equation}
for some large $C>0$, which gives the result.
\end{proof}

We next have the following:
\begin{proposition}[Mid-frequency bounds]\label{prop-resMF} Given (A1)-(A2), (H0)-(H2), and strong spectral stability (D), \begin{equation}
|(L_\txi - \lambda)^{-1}|_{\hat H^1(x_1)} \le C , \quad \mbox{for
}R^{-1}\le |(\txi,\lambda)|\le R \mbox{ and }\R\lambda \ge
-\theta,\end{equation} for any $R$ and $C=C(R)$ sufficiently large
and $\theta = \theta(R)>0$ sufficiently small, where $|\hat f|_{\hat
H^1(x_1)}$ is defined as in Proposition \ref{prop-resHF}.
\end{proposition}

\begin{proof}
Immediate, by compactness of the set of frequencies under
consideration together with the fact that the resolvent
$(\lambda-L_{\tilde{\xi}})^{-1}$ is analytic with respect to $H^{1}$
in $(\tilde{\xi}, \lambda)$;
see Proposition 4.8, \cite{Z4}.
\end{proof}

We next obtain the following resolvent bound for low-frequency
regions as a direct consequence of pointwise bounds on the resolvent
kernel, obtained in Proposition \ref{prop-Gbounds}.

\begin{proposition}[Low-frequency bounds]\label{prop-resLF}
Under the hypotheses of Theorem \ref{theo-nonlin},
for $\lambda \in \Gamma^{\tilde \xi}$ and $\rho :=|(\tilde
\xi,\lambda)|$, $\theta_1$ sufficiently small, there holds the
resolvent bound \begin{equation}\label{res-bound} |(L_{\tilde
\xi}-\lambda)^{-1}\partial_{x_1}^\beta \hat f|_{L^p(x_1)} \le
C\gamma_2\rho^{-2/p}\Big[\rho^{\beta}|\hat f|_{L^1(x_1)}
+ \beta|\hat f|_{L^\infty(x_1)} \Big]
\end{equation} for all $2\le p\le \infty$,
$\beta=0,1$, where $\gamma_2$ is as defined in
\eqref{gamma}.
\end{proposition}
\begin{proof} Using the convolution inequality
$|g*h|_{L^p}\le |g|_{L^p}|h|_{L^1}$
and noticing that
$$
|\partial_{y_1}^\beta G_{\tilde \xi,\lambda}(x_1,y_1)|
\le C \gamma_2(\rho^{\beta}+\beta e^{-\theta y_1})e^{-\theta\rho^2|x_1-y_1|},
$$
we obtain
\begin{equation}
\begin{aligned}
|(L_{\tilde \xi}-\lambda)^{-1}\partial_{x_1}^\beta \hat f|_{L^p(x_1)} &=
\Big|\int \partial_{y_1}^\beta G_{\tilde \xi,\lambda}(x_1,y_1)
\hat f(y_1,\tilde \xi)\, dy_1\Big|_{L^p(x_1)}\\
&\le
\Big|\int C \gamma_2
(\rho^{\beta}+\beta e^{-\theta y_1})e^{-\theta\rho^2|x_1-y_1|}
|\hat f(y_1, \tilde \xi)| \, dy_1\Big|_{L^p}\\
&\le C\gamma_2\rho^{-2/p}\Big[\rho^{\beta}|\hat f|_{L^1(x_1)}
+ \beta|\hat f|_{L^\infty(x_1)} \Big]
\end{aligned}
\end{equation}
as claimed.
\end{proof}

\begin{remark}
\textup{
The above $L^p$ bounds may alternatively be obtained directly
by the argument of Section 12, \cite{GMWZ1},
using quite different Kreiss symmetrizer techniques,
again omitting pole terms arising from vanishing of the Evans
function at the origin, and also the auxiliary problem
construction of Section 12.6 used to obtain sharpened bounds
in the Lax or overcompressive shock case (not relevant here).
}
\end{remark}

\subsection{Estimates on homogeneous solution operators}
Define low- and high-frequency parts of the linearized solution
operator $\cS(t)$ of the linearized problem with homogeneous
boundary and forcing data, $f$, $h\equiv 0$, as
\begin{equation}\cS_1(t):=\frac{1}{(2\pi i)^d}\int_{|\tilde \xi|\le
r}\oint_{\Gamma^{\tilde \xi}\cap \{|\lambda|\le r\}} e^{\lambda t + i\tilde \xi \cdot\tilde
x}(L_{\txi} - \lambda)^{-1} d\lambda d\txi\end{equation} and
\begin{equation} \cS_2(t):=e^{Lt} - \cS_1(t).\end{equation}

Then we obtain the following:
\begin{proposition}[Low-frequency estimate]\label{LFest}
Under the hypotheses of Theorem \ref{theo-nonlin},
for $\beta = (\beta_1,\beta')$ with $\beta_1=0,1$,
\begin{equation}
\begin{aligned} |\cS_1(t) \partial_x^\beta f|_{L^2_x} \le& C
(1+t)^{-(d-1)/4 - |\beta|/2}|f|_{L^1_x} +
C\beta_1(1+t)^{-(d-1)/4}|f|_{L^{1,\infty}_{\tx,x_1}} ,\\
| \cS_1(t) \partial_x^\beta f|_{L^{2,\infty}_{\tx,x_1}} \le& C
(1+t)^{-(d+1)/4 - |\beta|/2}|f|_{L^1_x} +
C\beta_1(1+t)^{-(d+1)/4}|f|_{L^{1,\infty}_{\tx,x_1}},\\
|\cS_1(t) \partial_x^\beta f|_{L^{\infty}_{\tx,x_1}} \le& C
(1+t)^{-d/2 - |\beta|/2}|f|_{L^1_x} +
C\beta_1(1+t)^{-d/2}|f|_{L^{1,\infty}_{\tx,x_1}},
\end{aligned}\end{equation}
where $|\cdot|_{L^{p,q}_{\tx,x_1}}$ denotes the norm in
$L^p(\tx;L^q(x_1))$.
\end{proposition}
\begin{proof}
The proof will follow closely the treatment of the shock case
in \cite{Z3}.
Let $\hat u(x_1,\txi,\lambda)$ denote the solution of
$(L_\txi-\lambda)\hat u = \hat f$, where $\hat f(x_1,\txi)$ denotes
Fourier transform of $f$, and
$$u(x,t):=\cS_1(t)f = \frac{1}{(2\pi i)^d}\int_{|\txi|\le r}\oint _{\Gamma^\txi\cap \{|\lambda|\le r\}}
e^{\lambda t+i\txi \cdot \tx}(L_\txi - \lambda)^{-1}\hat
f(x_1,\txi)d\lambda d\txi.$$

Recalling the resolvent estimates in Proposition \ref{prop-resLF},
we have
\begin{align*}|\hat u(x_1,\txi,\lambda)|_{L^p(x_1)}&\le C\gamma_2
\rho^{-2/p} |\hat f|_{L^1(x_1)}\le C\gamma_2 \rho^{-2/p} |f|_{L^1(x)}
\end{align*}
where $\gamma_2$ is as defined in \eqref{gamma}.

Therefore, using Parseval's identity, Fubini's theorem, and the triangle
inequality, we may estimate $$\begin{aligned}
|u|_{L^2(x_1,\tx)}^2(t) &= \frac{1}{(2\pi)^{2d}}\int_{x_1}
\int_{\txi}\Big|\oint_{\Gamma^\txi\cap \{|\lambda|\le r\}} e^{\lambda t}\hat
u(x_1,\txi,\lambda)d\lambda\Big|^2 d\txi dx_1
\\&=\frac{1}{(2\pi)^{2d}} \int_{\txi}\Big|\oint_{\Gamma^\txi\cap \{|\lambda|\le r\}}
e^{\lambda t}\hat u(x_1,\txi,\lambda)d\lambda\Big|^2_{L^2(x_1)}
d\txi \\&\le
\frac{1}{(2\pi)^{2d}}\int_{\txi}\Big|\oint_{\Gamma^\txi\cap \{|\lambda|\le r\}}
e^{\R\lambda t}|\hat u(x_1,\txi,\lambda)|_{L^2(x_1)}d\lambda\Big|^2
d\txi \\&\le
C|f|_{L^1(x)}^2\int_{\txi}\Big|\oint_{\Gamma^\txi\cap
\{|\lambda|\le r\}} e^{\R\lambda t}\gamma_2\rho^{-1}d\lambda\Big|^2
d\txi.
\end{aligned}$$

Specifically, parametrizing $\Gamma^\txi$ by $$\lambda(\txi,k) = ik
- \theta_1(k^2 + |\txi|^2), \quad k\in \RR,$$ and observing that
 by \eqref{gamma},
\begin{equation}\begin{aligned}\gamma_2\rho^{-1}&
\le(|k|+|\txi|)^{-1} \Big[ 1+
\sum_{j}\Big(\frac{|k-\tau_j(\txi)|}{\rho}\Big)^{1/s_j-1}\Big]\\&\le(|k|+|\txi|)^{-1}
\Big[ 1+
\sum_{j}\Big(\frac{|k-\tau_j(\txi)|}{\rho}\Big)^{\epsilon-1}\Big],\end{aligned}\end{equation}
where $\epsilon:=\frac{1}{\max_j s_j}$ ($0<\epsilon<1$ chosen
arbitrarily if there are no singularities), we estimate
$$\begin{aligned}
\int_{\txi}\Big|\oint_{\Gamma^\txi\cap \{|\lambda|\le r\}}
e^{\R\lambda t}\gamma_2\rho^{-1}d\lambda\Big|^2 d\txi &\le
\int_{\txi}\Big|\int_\RR e^{-\theta_1(k^2+|\txi|^2)
t}\gamma_2\rho^{-1}dk\Big|^2 d\txi\\&\le
\int_{\txi}e^{-2\theta_1|\txi|^2t}|\txi|^{-2\epsilon}\Big|\int_\RR
e^{-\theta_1k^2t}|k|^{\epsilon-1}dk\Big|^2 d\txi \\&\quad+\sum_j
\int_{\txi}e^{-2\theta_1|\txi|^2t}|\txi|^{-2\epsilon}\Big|\int_\RR
e^{-\theta_1k^2t}|k-\tau_j(\txi)|^{\epsilon-1}dk\Big|^2 d\txi
\\&\le
\int_{\txi}e^{-2\theta_1|\txi|^2t}|\txi|^{-2\epsilon}\Big|\int_\RR
e^{-\theta_1k^2t}|k|^{\epsilon-1}dk\Big|^2 d\txi \\&\le
Ct^{-(d-1)/2}.
\end{aligned}$$

Likewise, we have $$\begin{aligned}
|u|_{L^{2,\infty}_{\tx,x_1}}^2(t) &=
\frac{1}{(2\pi)^{2d}}\int_{\txi}\Big|\oint_{\Gamma^\txi \cap
\{|\lambda|\le r\}} e^{\lambda t}\hat
u(x_1,\txi,\lambda)d\lambda\Big|^2_{L^\infty(x_1)} d\txi \\&\le
\frac{1}{(2\pi)^{2d}}\int_{\txi}\Big|\oint_{\Gamma^\txi\cap
\{|\lambda|\le r\}} e^{\R\lambda t}|\hat
u(x_1,\txi,\lambda)|_{L^\infty(x_1)}d\lambda\Big|^2 d\txi
\\&\le
C|f|_{L^1(x)}^2\int_{\txi}\Big|\oint_{\Gamma^\txi\cap
\{|\lambda|\le r\}} e^{\R\lambda t}\gamma_2d\lambda\Big|^2 d\txi
\end{aligned}$$
where
$$\begin{aligned}
\int_{\txi}\Big|\oint_{\Gamma^\txi\cap \{|\lambda|\le r\}}
e^{\R\lambda t}\gamma_2d\lambda\Big|^2 d\txi &\le
\int_{\txi}e^{-2\theta_1|\txi|^2 t}\Big|\int_\RR
e^{-\theta_1k^2t}dk\Big|^2 d\txi
\\&\quad+\sum_j\int_{\txi}e^{-2\theta_1|\txi|^2 t}|\txi|^{2-2\epsilon}\Big|\int_\RR
e^{-\theta_1k^2t}|k-\tau_j(\txi)|^{\epsilon-1}dk\Big|^2 d\txi\\&\le
Ct^{-(d+1)/2} + C\int_{\txi}e^{-2\theta_1|\txi|^2
t}|\txi|^{2-2\epsilon}\Big|\int_\RR
e^{-\theta_1k^2t}|k|^{\epsilon-1}dk\Big|^2 d\txi\\&\le
Ct^{-(d+1)/2}.
\end{aligned}$$

Similarly, we estimate $$\begin{aligned} |u|_{L^\infty_{\tx,
x_1}}(t) &\le
\frac{1}{(2\pi)^{d}}\int_{\txi}\Big|\oint_{\Gamma^\txi \cap \{|\lambda|\le r\}} e^{\lambda
t}\hat u(x_1,\txi,\lambda)d\lambda\Big|_{L^\infty(x_1)} d\txi
\\&\le\frac{1}{(2\pi)^{d}} \int_{\txi}\oint_{\Gamma^\txi \cap \{|\lambda|\le r\}}
e^{\R\lambda t}|\hat u(x_1,\txi,\lambda)|_{L^\infty(x_1)}d\lambda
d\txi \\&\le
C|f|_{L^1(x)}\int_{\txi}\oint_{\Gamma^\txi \cap
\{|\lambda|\le r\}} e^{\R\lambda t}\gamma_2d\lambda d\txi
\end{aligned}$$
where as above we have
$$\begin{aligned}
\int_{\txi}\oint_{\Gamma^\txi\cap \{|\lambda|\le r\}} e^{\R\lambda
t}\gamma_2d\lambda d\txi &\le \int_{\txi}e^{-\theta_1|\txi|^2
t}\int_\RR e^{-\theta_1k^2t}dk d\txi
\\&\quad+\sum_j\int_{\txi}e^{-\theta_1|\txi|^2 t}|\txi|^{1-\epsilon}\int_\RR
e^{-\theta_1k^2t}|k-\tau_j(\txi)|^{\epsilon-1}dkd\txi\\&\le
Ct^{-d/2} + C\int_{\txi}e^{-\theta_1|\txi|^2
t}|\txi|^{1-\epsilon}\int_\RR
e^{-\theta_1k^2t}|k|^{\epsilon-1}dkd\txi\\&\le Ct^{-d/2}.
\end{aligned}$$

The $x_1$-derivative bounds follow similarly by using the resolvent
bounds in Proposition \ref{prop-resLF} with $\beta_1=1$. The
$\tx$-derivative bounds are straightforward by the fact that
$\widehat{\partial_{\tx}^{\tilde \beta} f} = (i\txi)^{\tilde \beta}
\hat f$.

Finally, each of the above integrals is bounded by
$ C|f|_{L^1(x)}$
as the product of
$ |f|_{L^1(x)}$
times the integral quantities
$\gamma_2 \rho^{-1}$, $\gamma_2$ over a bounded domain, hence we may
replace $t$ by $(1+t)$ in the above estimates.
\end{proof}

Next, we obtain estimates on the high-frequency part $\cS_2(t)$ of the
linearized solution operator.
Recall that $\cS_2(t) = \cS(t) - \cS_1(t)$, where
$$ \cS(t) = \frac{1}{(2\pi i)^{d}}\int_{\RR^{d-1}} e^{i\txi \cdot \tx} e^{L_{\txi}t}d\txi$$
and $$ \cS_1(t) = \frac{1}{(2\pi i)^{d}}\int_{|\txi|\le r}
\oint_{\Gamma^\txi \cap\{|\lambda|\le r\}} e^{\lambda t+i\txi \cdot
\tx}(L_\txi - \lambda)^{-1} d\lambda d\txi.$$

Then according to \cite[Corollary 4.11]{Z4}, we can write
\begin{equation}\label{formS2}\begin{aligned}\cS_2(t)f&=\frac{1}{(2\pi i)^{d}}\textup{P.V.}\int_{-\theta_1-i\infty}^{-\theta_1+i\infty}
\int_{\RR^{d-1}}\chi_{|\tilde{\xi}|^{2}+|\I\lambda|^{2}\geq\theta_1+\theta_2}\\
&\qquad\times e^{i\tilde{\xi}\cdot \tilde{x}+\lambda t}
(\lambda-L_{\tilde{\xi}})^{-1} \hat{f}(x_1,\tilde{\xi}) d\tilde{\xi}
d\lambda.
\end{aligned}\end{equation}

\begin{proposition}[High-frequency estimate] \label{prop-HFest}
Given (A1)-(A2), (H0)-(H2), (D), and homogeneous boundary conditions
(B), for $0\le |\alpha|\le s-3$, $s$ as in (H0),
\begin{equation}
\begin{aligned}|\cS_2(t)f|_{L^2_x} &\le C
e^{-\theta_1t}|f|_{H^{3}_x},\\|
\partial^\alpha_x\cS_2(t)f|_{L^{2}_{x}}&\le  C
e^{-\theta_1t}|f|_{H^{|\alpha|+3}_x}.\end{aligned}\end{equation}
\end{proposition}

\begin{proof} The proof starts with the following resolvent
identity, using analyticity on the resolvent set
$\rho(L_\txi)$ of the resolvent $(\lambda-L_\txi)^{-1}$, for all
$f\in \mathcal{D}(L_\txi)$,
\begin{equation}\label{res-id}
(\lambda-L_\txi)^{-1}f=\lambda^{-1}(\lambda-L_\txi)^{-1}L_\txi
f+\lambda^{-1}f.
\end{equation}

Using this identity and \eqref{formS2}, we estimate
\begin{equation}\label{S-est}\begin{aligned}\cS_2(t)f &=\frac{1}{(2\pi i)^{d}}
\textup{P.V.} \int_{-\theta_1-i\infty}^{-\theta_1+i\infty}
\int_{\RR^{d-1}}\chi_{|\tilde{\xi}|^{2}+|\I\lambda|^{2}\geq\theta_1+\theta_2}\\
&\qquad\qquad\times e^{i\txi\cdot\tx +\lambda
t}\lambda^{-1}(\lambda-L_\txi)^{-1}L_\txi\hat f(x_1,\txi) d\txi
d\lambda\\&\quad+\frac{1}{(2\pi i)^{d}}
\textup{P.V.}\int_{-\theta_1-i\infty}^{-\theta_1+i\infty}
\int_{\RR^{d-1}}\chi_{|\tilde{\xi}|^{2}+|\I\lambda|^{2}\geq\theta_1+\theta_2}\\
&\qquad\qquad\times e^{i\tilde{\xi}\cdot \tilde{x} +\lambda
t}\lambda^{-1}\hat f(x_1,\tilde{\xi}) d\tilde{\xi} d\lambda\\&=:S_1
+ S_2,
\end{aligned}\end{equation}
where, by Plancherel's identity and Propositions \ref{prop-HFest} and \ref{prop-resMF}, we have
$$\begin{aligned}
|S_1|_{L^2(\tx,x_1)}&\le C
\int_{-\theta_1-i\infty}^{-\theta_1+i\infty}
|\lambda|^{-1}|e^{\lambda
t}||(\lambda-L_\txi)^{-1}L_\txi\hat f|_{L^2(\txi,x_1)}|d\lambda|
\\&\le C e^{-\theta_1 t}
\int_{-\theta_1-i\infty}^{-\theta_1+i\infty}
|\lambda|^{-3/2}\Big|(1+|\txi|)|L_\txi\hat f|_{H^1(x_1)}\Big|_{L^2(\txi)}|d\lambda|
\\&\le C
e^{-\theta_1t}|f|_{H^{3}_x}
\end{aligned}$$
and\begin{equation}\begin{aligned}
|S_2|_{L^2_x}&\leq\frac{1}{(2\pi )^{d}}\Big|\text{P.V.}\int_{-\theta_1-i\infty}^{-\theta_1+i\infty}\lambda^{-1}e^{\lambda t} d\lambda \int_{\RR^{d-1}}e^{i\tilde{x}\cdot\tilde{\xi}}\hat f(x_1,\tilde{\xi}) d\tilde{\xi}\Big|_{L^2}\\
&\qquad+ \frac{1}{(2\pi )^{d}}\Big|\text{P.V.}\int_{-\theta_1-i r}^{-\theta_1+i r}\lambda^{-1}e^{\lambda t} d\lambda \int_{\RR^{d-1}}e^{i\tilde{x}\cdot\tilde{\xi}}\hat f(x_1,\tilde{\xi}) d\tilde{\xi}\Big|_{L^2}\\
&\leq Ce^{-\theta_1 t} |f|_{L^2_x},
\end{aligned}\end{equation}by direct computations, noting that the integral in $\lambda$ in the first term is identically zero.
This completes the proof of the first inequality stated in the proposition.
Derivative bounds follow similarly. \end{proof}

\begin{remark}\label{onehalf}
\textup{
Here, we have used the $\lambda^{1/2}$ improvement in \eqref{res-est}
over \eqref{oldres-est} together with modifications introduced in
\cite{KZ} to greatly simplify the original high-frequency
argument given in \cite{Z3} for the shock case.
}
\end{remark}

\subsection{Boundary estimates} For the purpose of studying the nonzero boundary perturbation, we need the following proposition.
For $h:=h(\tilde x,t)$, define
\begin{equation}\label{Dh}\cD_h(t):= (|h_t| + |h_{\tx}| + |h_{\tx\tx}|)(t),\end{equation} and
\begin{equation}\label{Gammah}
\Gamma h(t):=
\begin{aligned}
\int_0^t\int_{\RR^{d-1}}
\Big(\sum_kG_{y_k}B^{k1}&+GA^1\Big)
(x,t-s;0,\tilde y)h(\tilde y,s)\,d\tilde yds,
\end{aligned}
\end{equation}
where $G(x,t;y)$ is the Green function of $\partial_t - L$.
This boundary term will appear when we write down the Duhamel
formulas for the linearized and nonlinear equations (see
\eqref{Duhamel-lin} and \eqref{Duhamel}).
Noting that for the outflow case, the fact that $G(x,t;0,\tilde y)\equiv 0$ simplifies $\Gamma h$ to \begin{equation}\label{Bh1-out}\Gamma h(t)=\int_0^t\int_{\RR^{d-1}}G_{y_1}(x,t-s;0,\tilde y)B^{11}h\,d\tilde yds.\end{equation}
Therefore when dealing with the outflow case, instead of putting assumptions on $h$ itself as in the inflow case,
we make assumptions on $B^{11}h$,
matching with the hypotheses on $W$-coordinates.

\begin{prop}\label{prop-BCs-est} Assume that $h=h(\tx,t)$ satisfies \begin{equation}\label{hdecay}\begin{aligned}|h(t)|_{L^2_{\tx}}&\le E_0(1+t)^{-(d+1)/4}, \\ |h(t)|_{L^\infty_{\tx}}&\le E_0(1+t)^{-d/2}\\ |\cD_h(t)|_{L^1_{\tx} \cap H^{|\gamma|+3}_{\tx}} &\le E_0(1+t)^{-d/2 - \epsilon},\end{aligned}\end{equation}for some positive constant $E_0$; here 
$|\gamma|=[(d-1)/2]+2$, 
and $\epsilon>0$ is arbitrary small for $d=2$ and $\epsilon=0$ for $d\ge 3$. For the outflow case, we replace these assumptions on $h$ by those on $B^{11}h$.
Then we obtain
\begin{equation}\begin{aligned}|\Gamma h(t)|_{L^2} &\le CE_0(1+t)^{-(d-1)/4}, \\|\Gamma h(t)|_{L^{2,\infty}_{\tx,x_1}} &\le CE_0(1+t)^{-(d+1)/4},
\\|\Gamma h(t)|_{L^\infty}&\le CE_0(1+t)^{-d/2},\end{aligned}\end{equation} and derivative bounds
\begin{equation}\begin{aligned}|\partial_x\Gamma h(t)|_{L^{2,\infty}_{\tx,x_1}}&\le CE_0(1+t)^{-(d+1)/4},
\\|\partial_{\tx}^2\Gamma h(t)|_{L^{2,\infty}_{\tx,x_1}}&\le CE_0(1+t)^{-(d+1)/4},\end{aligned}\end{equation} for all $t\ge 0$.
\end{prop}

\begin{proof} We first recall that $G(x,t-s;y)$ is a solution of $(\partial_s-L_y)^*G^* =0$, that is,  \begin{equation}\label{explicit}
-G_s -\sum_j(GA^j)_{y_j} + \sum_jGA^j_{y_j}= \sum_{jk}(G_{y_k}B^{kj})_{y_j}.
\end{equation}
Integrating this on $\RR^{d}_+\times [0,t]$ against
\begin{equation}\label{gdef}
g(y_1,\tilde y,s):=e^{-y_1}h(\tilde y, s),
\end{equation}
and integrating by parts twice, we obtain
$$\begin{aligned} \Gamma h&= -
\int_0^t\int_{\RR^{d}_+} \Big(\sum_{jk}G_{y_k}B^{kj} + \sum_jGA^j\Big)g_{y_j} dy ds \\&\quad -\int_0^t\int_{\RR^{d}_+}\Big(-G_s + \sum_j GA^j_{y_j}\Big)g(y,s)dyds
\end{aligned}$$
where, recalling that
$$ \cS(t)f(x) = \int_{\RR^{d}_+} G(x,t;y) f(y)dy,$$
we get
$$\begin{aligned}- \int_0^t\int_{\RR^{d}_+} &\sum_{jk}\Big(G_{y_k}B^{kj} + \sum_jGA^j\Big)g_{y_j} dy ds \\&= -\int_0^t \cS(t-s)\Big(-\sum_{jk}(B^{kj}g_{x_j})_{x_k} + \sum_j A^jg_{x_j}\Big)ds\end{aligned}$$
and $$\begin{aligned} -\int_0^t\int_{\RR^{d}_+}&\Big(-G_s + \sum_j GA^j_{y_j}\Big)g(y,s)dyds \\&= -\int_0^t \cS(t-s)\Big(g_s + \sum_j A^j_{x_j}g\Big)ds + g(x,t)- \cS(t)g(x,0).
\end{aligned}$$
Therefore combining all these estimates yields \begin{equation}\label{Bh2}\Gamma h = g(x,t)- \cS(t)g_0-\int_0^t \cS(t-s)(g_s-L_xg(x,s))ds\end{equation} with $g_0(x):=g(x,0)$
and $L_xg = -\sum_{j}(A^jg)_{x_j} + \sum_{jk}(B^{jk}g_{x_k} )_{x_j}.$

Now we are ready to employ estimates obtained in the previous section
on the solution operator $\cS(t) = \cS_1(t)+\cS_2(t)$.
Noting that
$$
|g|_{L^p_{x}} \le C |h|_{L^p_{\tx}},
$$
we estimate
$$\begin{aligned} |\Gamma h|_{L^2} &\le
|g|_{L^2}+|\cS_1(t)g_0|_{L^2} + |\cS_2(t)g_0|_{L^2}\\&\quad + \int_0^t |\cS_1(t-s)(g_s-Lg)|_{L^2} + |\cS_2(t-s)(g_s-Lg)|_{L^2} ds\\&\le
|h(t)|_{L^2_{\tx}}
+C(1+t)^{-\frac{d-1}{4}}|g_0|_{L^1}
 + Ce^{-\eta
t}|g_0|_{H^3}\\&\quad + \int_0^t
(1+t-s)^{-(d-1)/4}(|g_s|+|Lg|)_{L^1} +
e^{-\theta(t-s)}(|g_s|+|Lg|)_{H^3} ds\\&\le
|h(t)|_{L^2_{\tx}}+C(1+t)^{-\frac{d-1}{4}}|h_0|_{L^1_{\tx}\cap
H^3_{\tx}}\\&\quad+  \int_0^t
(1+t-s)^{-(d-1)/4}|\cD_h(s)|_{L^1_{\tx}} +
e^{-\theta(t-s)}|\cD_h(s)|_{H^3_{\tx}} ds\\&\le
CE_0(1+t)^{-\frac{d-1}{4}}
\end{aligned}$$
and similarly we also obtain
\begin{equation} \begin{aligned} |\Gamma h|_{L^{2,\infty}_{\tx,x_1}}&\le
|h(t)|_{L^{2}_{\tx}}+C(1+t)^{-\frac{d+1}{4}}|h_0|_{L^1_{\tx}\cap H^4_{\tx}}\\&\quad+ C\int_0^t (1+t-s)^{-(d+1)/4}|\cD_h(s)|_{L^1_{\tx}} + e^{-\theta(t-s)}|\cD_h(s)|_{H^4_{\tx}} ds\\&\le
CE_0(1+t)^{-\frac{d+1}{4}}
\end{aligned}\end{equation} and
\begin{equation} \begin{aligned} |\Gamma h|_{L^{\infty}}&\le
|h(t)|_{L^{\infty}_{\tx}}+C(1+t)^{-\frac{d}{2}}|h_0|_{L^1_{\tx}\cap H^{|\gamma|+3}_{\tx}}\\&\quad+ C\int_0^t (1+t-s)^{-d/2}|\cD_h(s)|_{L^1_{\tx}} + e^{-\theta(t-s)}|\cD_h(s)|_{H^{|\gamma|+3}_{\tx}} ds\\&\le
CE_0(1+t)^{-\frac{d}{2}}.
\end{aligned}\end{equation} Similar bounds hold for derivatives.

This completes the proof of the proposition.
\end{proof}

\subsection{Duhamel formula}
The following integral representation formula
expresses the solution of the inhomogeneous
equation \eqref{inhom} in terms of the homogeneous solution
operator $\cS$ for $f$, $h\equiv 0$.

\begin{lemma}[Integral formulation]\label{lem-duhamel}
Solutions $U$ of \eqref{inhom} may be expressed as
\begin{equation}\label{Duhamel1}
\begin{aligned}
  U(x,t)=& \cS(t) U_0 + \int_0^t \cS(t-s) f(\cdot, s)
+ \Gamma U(0,\tx,t)
\end{aligned}
\end{equation} where
$U(x,0) = U_0(x),$
\begin{equation}\label{GammaU1}\Gamma U(0,\tx,t):=\int_0^t\int_{\RR^{d-1}}(\sum_jG_{y_j}B^{j1}+GA^1)(x,t-s;0,\tilde y)U(0,\tilde y,s)\,d\tilde yds,
\end{equation}
and
$G(\cdot, t;y)=\cS(t)\delta_y(\cdot)$
is the Green function of $\partial_t - L$.
\end{lemma}

\begin{proof}
Integrating on $\RR^d_+$ the linearized equations
$$
(\partial_s - L_y)U = f
$$
against $G(x,t-s;y)$ and using the fact that by duality
$$
(\partial_s - L_y)^* G^* (x,t-s;y) \equiv 0,
$$
we easily obtain the lemma as in the one-dimensional case (see \cite{YZ,NZ}),
 recalling that $$ \cS(t)f = \int_{\RR^{d}_+}
G(x,t;y) f(y)dy.$$\end{proof}

\subsection{Proof of linearized stability}
\begin{proof}[Proof of Theorem \ref{theo-lin}]
Writing the Duhamel formula for the linearized equations \begin{equation}\label{Duhamel-lin}
\begin{aligned}
  U(x,t)=& \cS(t) U_0 + \Gamma h(\tx,t),
\end{aligned}
\end{equation} with $\Gamma h$ defined in \eqref{Gammah}, where $U(x,0)=U_0(x)$ and $U(0,\tx,t) = h(\tx,t)$, and applying estimates on low- and high-frequency operators $\cS_1(t)$
and $\cS_2(t)$, we obtain \begin{equation} \begin{aligned} |U(t)|_{L^2} &\le
|\cS_1(t)U_0|_{L^2} + |\cS_2(t)U_0|_{L^2} +|\Gamma h(t)|_{L^2}\\&\le
C(1+t)^{-\frac{d-1}{4}}|U_0|_{L^1} + Ce^{-\eta t}|U_0|_{H^3}+CE_0(1+t)^{-(d-1)/4}\\&\le
C(1+t)^{-\frac{d-1}{4}}(|U_0|_{L^1\cap H^3}+E_0)
\end{aligned}\end{equation}
and \begin{equation} \begin{aligned} |U(t)|_{L^\infty} &\le
|\cS_1(t)U_0|_{L^\infty} + |\cS_2(t)U_0|_{L^\infty} +|\Gamma h(t)|_{L^\infty}\\&\le
C(1+t)^{-\frac{d}{2}}|U_0|_{L^1} +
C|\cS_2(t)
U_0|_{H^{[(d-1)/2]+2}}
+CE_0(1+t)^{-d/2}\\&\le C(1+t)^{-\frac{d}{2}}|U_0|_{L^1}
+ Ce^{-\eta t}
|U_0|_{H^{[(d-1)/2]+5}}+CE_0(1+t)^{-d/2}\\&\le
C(1+t)^{-\frac{d}{2}}(
|U_0|_{L^1\cap H^{[(d-1)/2]+5}}+E_0).
\end{aligned}\end{equation}
These prove the bounds as stated in the theorem for $p=2$ and
$p=\infty$. For $2<p<\infty$, we use the interpolation inequality
between $L^2$ and $L^\infty$.
\end{proof}

\section{Nonlinear stability}

\subsection{Auxiliary energy estimates}\label{sec-EE}
For the analysis of nonlinear stability, we need the following energy estimate adapted from
\cite{MaZ4,NZ,Z4}. Define
the nonlinear perturbation variables $U = (u, v)$ by
\begin{equation}\label{per-var}
U(x,t):=\tilde U(x,t)-\bU(x_1).
\end{equation}

\begin{proposition}\label{prop-energy-est}
Under the hypotheses of Theorem \ref{theo-nonlin}, let $U_0
\in H^s$ and $U=(u,v)^T$ be a solution of \eqref{sys} and
\eqref{per-var}. Suppose that, for $0\le t\le T$, the
$W^{2,\infty}_x$ norm of the solution $U$ remains bounded by a
sufficiently small constant $\zeta>0$. Then
\begin{align}\label{energy-ineq}|U(t)|_{H^s}^2 \le Ce^{-\theta t}|U_0|_{H^s}^2 +
C \int_0^t
e^{-\theta(t-\tau)}\Big(|U(\tau)|_{L^2}^2+|\CalB_h(\tau)|^2\Big)d\tau\end{align}
for all $0\le t\le T$, where the boundary term $\CalB_h$ is
defined as in Theorem \ref{theo-nonlin}.
\end{proposition}

\begin{proof}Observe that a straightforward calculation shows that
$|U|_{H^r}\sim|W|_{H^r},$
\begin{equation}\label{per-varW}W = \tilde W - \bar W := W(\tilde U ) -W(\bU),\end{equation}
for $0\le r\le s$, provided $|U|_{W^{2,\infty}}$ remains bounded, hence it is sufficient to prove a
corresponding bound in the special variable $W$. We first carry out a complete proof
in the more straightforward case with conditions (A1)-(A3) replaced by the following global
versions, indicating afterward by a few brief remarks the changes needed to carry
out the proof in the general case.
\medskip
~\\
(A1') $\tilde A^j(\tilde W),\tilde A^0,\tilde A^1_{11}$ are
symmetric, $\tilde A^0\ge \theta_0>0$,
\medskip
~\\
(A2') Same as (A2),
\medskip
~\\
(A3') $\tilde W = \begin{pmatrix}\tilde w^I\\\tilde w^{II}\end{pmatrix},
\quad  \tilde B^{jk}=\tilde B^{kj}=\begin{pmatrix}0 & 0 \\
0 & \tilde b^{jk}\end{pmatrix},\quad \sum \xi_j\xi_k\tilde b^{jk}\ge
\theta|\xi|^2,$ and $\tilde G\equiv 0$.

Substituting \eqref{per-varW} into \eqref{symmetric}, we obtain the quasilinear perturbation equation
\begin{align}\label{perturb-eqs} A^0 W_t + \sum_jA^jW_{x_j} =
\sum_{jk}(B^{jk}W_{x_k})_{x_j} + M_1\bW_{x_1} +
\sum_j(M_2^j\bW_{x_1})_{x_j}\end{align} where $A^0:=A^0(W+\bW)$ is
symmetric positive definite, $A^j:=A^j(W+\bW) $ are symmetric,
\begin{align*}M_1 &= A^1(W+\bW) - A^1(\bW) = \Big(\int_0^1 dA^1(\bW +
\theta
W)d\theta\Big)W,\\
M_2^j &= B^{j1}(W+\bW) - B^{j1}(\bW) = \begin{pmatrix}0 & 0 \\
0 & (\int_0^1 db^{j1}(\bW + \theta
W)d\theta)W\end{pmatrix}.\end{align*}

As shown in \cite{MaZ4}, we have bounds
\begin{align}\label{bound-A1}|A^0|\le C,\quad |A^0_t| &\le C|W_t|\le C(|W_x|+|w^{II}_{xx}|)\le
C\zeta,\\|\partial_xA^0|+|\partial_x^2A^0|&\le
C(\sum_{k=1}^2|\partial_x^kW|+|\bW_{x_1}|)\le
C(\zeta+|\bW_{x_1}|).\label{bound-A2}\end{align} We have the same
bounds for $A^j$, $B^{jk}$, and also due to the form of $M_1,M_2$,
\begin{align}\label{bound-M}|M_1|,|M_2|\le C(\zeta+|\bW_{x_1}|)|W|. \end{align}

Note that thanks to Lemma \ref{lem-profile-decay} we have the bound
on the profile: $|\bW_{x_1}|\le Ce^{-\theta |x_1|}$, as $x_1 \to
+\infty$.

The following results assert that hyperbolic effects can compensate
for degenerate viscosity $B$, as revealed by the existence of a
compensating matrix $K$.

\begin{lemma}[\cite{KSh}]\label{K} Assuming (A1'), condition (A2') is equivalent to the following:

(K1) There exist smooth skew-symmetric ``compensating matrices''
$K(\xi)$, homogeneous degree one in $\xi$, such that
\begin{equation}\label{K1}\R\Big(\sum_{j,k}\xi_j\xi_kB^{jk}-K(\xi)(A^0)^{-1}\sum_k\xi_kA^k \Big)(W_+)\ge
\theta_2|\xi|^2>0\end{equation}for all $\xi \in
\RR^d\setminus\{0\}$.
\end{lemma}

Define $\alpha$ by the ODE \begin{equation}\label{alphaeq}
\alpha_{x_1} = -\mbox{sign}(A^1_{11})c_*|\bW_{x_1}|\alpha, \quad
\alpha(0)=1\end{equation} where $c_*>0$ is a large constant to be
chosen later. Note that we have \begin{align} \label{alpha-est}
(\alpha_{x_1}/\alpha)A^1_{11} \le
-c_*\theta_1|\bW_{x_1}|=:-\omega(x_1)\end{align}
 and  \begin{align} \label{alpha-est1}
|\alpha_{x_1}/\alpha |\le
c_*|\bW_{x_1}|=\theta_1^{-1}\omega(x_1).\end{align}

In what follows, we shall use $\wprod{\cdot,\cdot}$ as the
$\alpha$-weighted $L^2$ inner product defined as $$\wprod{f,g} =
\wprod{\alpha f,g}_{L^2(\RR^d_+)}$$ and $$\|f\|_s =
\sum_{i=0}^s\sum_{|\alpha|=i}\Big<\partial_x^\alpha
f,\partial_x^\alpha f\Big>^{1/2}$$ as the norm in weighted $H^s$
space. Note that for any symmetric operator $S$,
\begin{align*}\wprod{S f_{x_j},f} &= -\frac 12\wprod{S_{x_j}f,f}, \quad j\not=1\\
\wprod{S f_{x_1},f} &= -\frac
12\wprod{(S_{x_1}+(\alpha_{x_1}/\alpha)S)f,f}-\frac 12 \iprod{
Sf,f},\end{align*} where $\iprod{\cdot,\cdot}$ denotes the
integration on $\RR^d_0:=\{x_1=0\}\times \RR^{d-1}$. Also we define
$$\|f\|_{0,s} = \|f\|_{H^s(\RR^d_0)}=
\sum_{i=0}^s\sum_{|\alpha|=i}\Big<\partial^\alpha_{\tx}f,\partial^\alpha_{\tx}f\Big>_0^{1/2}.$$

Note that in what follows, we shall pay attention to keeping track
of $c_*$. For constants independent of $c_*$, we simply write them
as $C$. Also, for simplicity,
the sum symbol will sometimes be dropped where it is no confusion.
We write $\|f_x\| = \sum_j\|f_{x_j}\|$ and
$\|\partial_x^kf\| = \sum_{|\alpha|=k}\|\partial_x^\alpha f\|$.

\subsubsection{Zeroth order ``Friedrichs-type'' estimate} First,
by integration by parts and estimates
\eqref{bound-A1}, \eqref{bound-A2}, and then \eqref{alpha-est}, we
obtain for $j\not=1$,
\begin{align*}-\langle A^j&W_{x_j},W\rangle = \frac
12\wprod{A^j_{x_j}W,W} \le C\wprod{(\zeta+|\bW_{x_1}|)w^{I},w^{I}}+
C\|w^{II}\|^2_0
\end{align*} and for $j=1$,
\begin{align*}
-\langle A^1W_{x_1},W\rangle &= \frac 12
\wprod{(A^1_{x_1}+(\alpha_{x_1}/\alpha)A^1)W,W}+ \frac 12
\iprod{A^1W, W}\\&\le \frac 12
\wprod{(\alpha_{x_1}/\alpha)A^1_{11}w^{I},w^{I}} +
C\wprod{(\zeta+|\bW_{x_1}|)|W| + \omega(x_1)|w^{II}|,|W|} + J^0_b
\\&\le
-\frac 12\wprod{\omega (x)w^{I},w^{I}}+
C\wprod{(\zeta+|\bW_{x_1}|)w^{I},w^{I}}+ C(c_*)\|w^{II}\|^2_0+
J^0_b,
\end{align*} where $J^0_b$ denotes the boundary term $\frac 12 \iprod{A^1W, W}$. The term $\wprod{|\bW_{x_1}|w^{I},w^{I}}$ may be easily absorbed into
the first term of the right-hand side, since for $c_*$ sufficiently large,
\begin{equation}\label{est-v}\wprod{|\bW_{x_1}|w^{I},w^{I}} \le (c_*\theta_1)^{-1}\wprod{\omega(x_1)w^{I},w^{I}}\le \frac 1{4C}\wprod{\omega(x_1)w^{I},w^{I}}.\end{equation}
Also, integration by parts yields
\begin{align*}
\langle (B^{jk}W_{x_k})_{x_j},W\rangle&=
-\wprod{B^{jk}W_{x_k},W_{x_j}}-\wprod{(\alpha_{x_1}/\alpha)B^{1k}W_{x_k},W} - \iprod{B^{1k}W_{x_k},W}
\\&\le
-\theta
\|w^{II}_{x}\|_0^2+C\wprod{\omega(x_1)w^{II}_{x},w^{II}}-\iprod{b^{1k}w^{II}_{x_k},w^{II}}\\&\le
-\theta \|w^{II}_{x}\|_0^2+C(c_*)\|w^{II}\|_0^2 -
\iprod{b^{1k}w^{II}_{x_k},w^{II}}.
\end{align*}
where we used the fact that $B^{jk}W_{x}\cdot W =
b^{jk}w^{II}_{x}\cdot w^{II}$,
noting that $B$ has block-diagonal form with the first block identical to zero. Similarly, recalling that
$M^j_2 = B^{j1}(W+\bW) -
B^{j1}(\bW)$, we have
\begin{align*}
\langle (M^j_2\bW_{x_1})_{x_j},W\rangle&=
-\wprod{M^j_2\bW_{x_1},W_{x_j}}-\wprod{(\alpha_{x_1}/\alpha)M^1_2\bW_{x_1},W
}- \iprod{M^1_2\bW_{x_1},W}\\&\le
C\wprod{|\bW_{x_1}||W|,|w^{II}_{x}|} +
C\wprod{\omega(x_1)|W|,w^{II}}- \iprod{m^1_2\bW_{x_1},w^{II}}
\\&\le \xi\|w^{II}_{x}\|_0^2 + C\Big(\epsilon\wprod{\omega(x_1)w^{I},w^{I}}+ C(c_*)\|w^{II}\|_0^2\Big)- \iprod{m^1_2\bW_{x_1},w^{II}}
\end{align*}for any small $\xi,\epsilon$. Note that $C$ is independent of $c_*$. Therefore, for $\xi=\theta/2$ and $c_*$ sufficiently large,
combining all above estimates, we obtain
\begin{equation}\label{Fzeroth-est0}\begin{aligned}
\frac 12 \dt\wprod{A^0W,W}
 &= \wprod{A^0W_t,W} +
\frac 12 \wprod{A^0_tW,W}\\
&=\wprod{-A^jW_{x_j} + (B^{jk}W_{x_k})_{x_j}+M_1\bW_{x_1} +
(M^j_2\bW_{x_1})_{x_j},W}+\frac 12
\wprod{A^0_tW,W}\\
&\le -\frac 14[\wprod{\omega(x_1)w^{I},w^{I}} +\theta
\|w^{II}_{x}\|_0^2]+ C\zeta\|w^{I}\|_0^2 + C(c_*)\|w^{II}\|^2_0 +
I_b^0
\end{aligned}\end{equation} where the boundary term
\begin{equation} I_b^0:=\frac 12 \iprod{A^1W,W}-
\iprod{b^{1k}w^{II}_{x_k},w^{II}} -
\iprod{m^1_2\bW_{x_1},w^{II}}\end{equation}
which, in the outflow case (thanks to the
negative definiteness of $A^1_{11}$), is estimated as
\begin{equation} \label{Ib0-out}I_b^0\le -\frac {\theta_1}{2}\|w^{I}\|_{0,0}^2 + C(\|w^{II}\|_{0,0}^2 +
\|w^{II}_{x}\|_{0,0}\|w^{II}\|_{0,0}),\end{equation}
and similarly in the inflow case, estimated as
\begin{equation} \label{Ib0-in}I_b^0\le C(\|W\|_{0,0}^2
+ \|w^{II}_{x}\|_{0,0}\|w^{II}\|_{0,0}).\end{equation}
Here we recall
that $\|\cdot\|_{0,s}:= \|\cdot\|_{H^s(\RR^d_0)}$.

\subsubsection{First order ``Friedrichs-type'' estimate} Similarly as
above, we need the following key estimate, computing by the
use of integration by parts, \eqref{est-v}, and $c_*$ being
sufficiently large,
\begin{equation}\label{ineq-key1}\begin{aligned}
-\sum_j&\wprod{W_{x_i},A^jW_{x_ix_j} }
\\&=\frac12\sum_{j}
\wprod{W_{x_i},A^j_{x_j}W_{x_i}}+\frac12
\wprod{W_{x_i},(\alpha_{x_1}/\alpha)A^1W_{x_i}} +\frac12
\iprod{W_{x_i},A^1W_{x_i}} \\&\le -\frac
14\wprod{\omega(x_1)w^{I}_{x},w^{I}_{x}}+C\zeta\|w^{I}_{x}\|_0^2
+Cc_*^2\|w^{II}_{x}\|_0^2+\frac12 \iprod{W_{x_i},A^1W_{x_i}}.
\end{aligned}\end{equation}
We deal with the boundary term later. Now let us compute
\begin{align}\label{F1} \frac 12 \dt &\wprod{A^0W_{x_i},W_{x_i}}
=\wprod{W_{x_i},(A^0W_t)_{x_i}}-\wprod{W_{x_i},A^0_{x_i}W_t}
+\frac12 \wprod{A^0_tW_{x_i},W_{x_i}}.
\end{align}
We control each term in turn. By \eqref{bound-A1} and \eqref{bound-A2}, we first have
\begin{align*}\wprod{A^0_tW_{x_i},W_{x_i}} \le C\zeta\|W_{x}\|_0^2
\end{align*} and by multiplying $(A^0)^{-1}$ into
\eqref{perturb-eqs}, \begin{align*} |\wprod{W_{x_i},A^0_{x_i}W_t}|
\le& C\wprod{(\zeta+|\bW_{x_1}|)|W_{x}|,
(|W_{x}|+|w^{II}_{xx}|+|W|)}\\\le& \xi\|w^{II}_{xx}\|_0^2 +
C\wprod{(\zeta+|\bW_{x_1}|)w^{I}_{x},w^{I}_{x}}+
C\wprod{(\zeta+|\bW_{x_1}|)w^{I},w^{I}} + C\|w^{II}\|_1^2,
\end{align*}where
the term $\wprod{|\bW_{x_1}|w^{I}_{x},w^{I}_{x}}$ may be treated in the same way
as was $\wprod{|\bW_{x_1}|w^{I},w^{I}}$ in \eqref{est-v}.
Using \eqref{perturb-eqs},
we write the first term in the right-hand side of \eqref{F1} as
\begin{align*}\wprod{W_{x_i},(A^0W_t)_{x_i}}=&\wprod{W_{x_i},[-A^jW_{x_j} +
(B^{jk}W_{x_k})_{x_j}+M_1\bW_{x_1}+(M^j_2\bW_{x_1})_{x_j}]_{x_i}}\\=&-\wprod{W_{x_i},A^jW_{x_ix_j}}
+\wprod{W_{x_i},-A^j_{x_i}W_{x_j}+(M_1\bW_{x_1})_{x_i}}
\\&-\wprod{W_{x_ix_j},[(B^{jk}W_{x_k})_{x_i}+(M_2^j\bW_{x_1})_{x_i}]}
\\&-\wprod{(\alpha_{x_1}/\alpha)W_{x_i},[(B^{1k}W_{x_k})_{x_i}+(M_2^1\bW_{x_1})_{x_i}]}
\\&-\iprod{W_{x_i},[(B^{1k}W_{x_k})_{x_i}+(M_2^1\bW_{x_1})_{x_i}]}
\\\le&-\frac 14\Big[\wprod{\omega(x_1)w^{I}_{x},w^{I}_{x}}+\theta\|w^{II}_{xx}\|_0^2\Big]\\&+C\Big[\zeta\|w^{I}\|_1^2
+C(c_*)\|w^{II}_{x}\|_0^2+\wprod{|\bW_{x_1}|w^{I},w^{I}}\Big]+I_b^1\end{align*}
where $I_b^1$ denotes the boundary terms
\begin{equation}\label{Ib1}\begin{aligned}I_b^1:&=\frac12
\iprod{W_{x_i},A^1W_{x_i}}
-\iprod{W_{x_i},[(B^{1k}W_{x_k})_{x_i}+(M_2^1\bW_{x_1})_{x_i}]}\\&=\frac12
\iprod{W_{x_i},A^1W_{x_i}}
-\iprod{w^{II}_{x_i},[(b^{1k}w^{II}_{x_k})_{x_i}+(m_2^1\bW_{x_1})_{x_i}]},\end{aligned}\end{equation}
and we have used (A3) for each fixed $i$ and $\xi_j =
(W_{x_i})_{x_j}$ to get
\begin{equation}\sum_{jk}\wprod{W_{x_ix_j},B^{jk}W_{x_kx_i}} \ge \theta
\sum_j\|W_{x_ix_j}\|^2_0 ,\end{equation} and estimates
\eqref{ineq-key1},\eqref{est-v} for $w^{I},w^{I}_{x}$,
and Young's inequality to obtain:
\begin{align*}
\wprod{W_{x},-A^j_{x}W_{x}+(M_1\bW_{x_1})_{x}}&\le
C\wprod{(\zeta+|\bW_{x_1}|)|W_{x}|,|W_{x}|+|W|}.
 \\
-\wprod{W_{xx}+(\alpha_{x_1}/\alpha)W_{x},(B^{jk}W_{x})_{x}}&\le\\
 -\theta\|w^{II}_{xx}\|_0^2 &+ C\langle
 |w^{II}_{xx}|+\omega(x_1)|w^{II}_{x}|,(\zeta+|\bW_{x_1}|)|w^{II}_{x}|\rangle
\\
-\wprod{W_{xx}+(\alpha_{x_1}/\alpha)W_{x},(M_2^j\bW_{x_1})_{x}}&\le\\
C\langle
|w^{II}_{xx}|&+\omega(x_1)|w^{II}_{x}|,(\zeta+|\bW_{x_1}|)(|W_{x}|+|W|)\rangle.
\end{align*}

Putting these estimates together into \eqref{F1}, we have obtained
\begin{align} \frac 12 \dt \wprod{A^0W_{x},W_{x}}  &+\frac 14\theta\|w^{II}_{xx}\|^2_{0}+\frac 14
\wprod{\omega(x_1)w^{I}_{x},w^{I}_{x}}\notag\\&\le
C\Big[\zeta\|w^{I}\|_1^2
+\wprod{|\bW_{x_1}|w^{I},w^{I}}+C(c_*)\|w^{II}\|_1^2\Big]+I_b^1.\label{ineq-key2}
\end{align}
Let us now treat the boundary term. First observe that using the
parabolic equations, noting that $A^0$ is the diagonal-block form,
we can estimate
\begin{align}\label{BCest1}(b^{jk}w^{II}_{x_k})_{x_j}(0,\tx,t)\le C\Big(|w^{II}_t|  +|W_{x_j}| + |W|
\Big)(0,\tx,t)
\end{align}
and thus for $i\not=1$
\begin{equation}\label{BCest2}\begin{aligned}\langle w^{II}_{x_i},[(b^{1k}w^{II}_{x_k})_{x_i}&+(m_2^1\bW_{x_1})_{x_i}]\rangle_0\\&\le \int_{\RR^d_0}|w^{II}_{x_ix_i}|\big(|W|+|w^{II}_{x_k}|\Big)
\\&\le C\int_{\RR^d_0}\Big(|W|^2+|w^{II}_{x}|^2+|w^{II}_{\tx \tx}|^2\Big)
\end{aligned}\end{equation}
and for $i=1$, using $b^{1k} = b^{k1}$, \eqref{BCest1}, and recalling here that we always use the sum convention,
\begin{equation}\label{BCest3}\begin{aligned}\sum_k(b^{1k}w^{II}_{x_k})_{x_1}
&=\frac 12\Big((b^{1k}w^{II}_{x_k})_{x_1}+(b^{j1}w^{II}_{x_1})_{x_j} + b_{x_1}^{1k}w^{II}_{x_k} - b_{x_j}^{j1}w^{II}_{x_1}\Big)
\\& =\frac 12\Big((b^{jk}w^{II}_{x_k})_{x_j}+ b_{x_1}^{1k}w^{II}_{x_k} - b_{x_j}^{j1}w^{II}_{x_1}-\sum_{j\not=1;\,k\not=1}(b^{jk}w^{II}_{x_k})_{x_j} \Big)
\\&
\le C\big(|w^{II}_t|  +|W|+|W_{x_j}|  + |w^{II}_{\tx\tx}|\Big)
.\end{aligned}\end{equation}
Therefore \begin{align*}\langle w^{II}_{x_1},&[(b^{1k}w^{II}_{x_k})_{x_1}+(m_2^1\bW_{x_1})_{x_1}]\rangle_0
\\&\le\epsilon\int_{\RR^d_0}|w^{I}_x|^2+C\int_{\RR^d_0}\Big(|w^{II}_t|^2+|W|^2+|w^{II}_{x}|^2+|w^{II}_{\tx\tx}|^2\Big)
\end{align*}

For the first term in $I^1_b$,
we consider each inflow/outflow case separately.
For the outflow case, since $A^1_{11}\le -\theta_1<0$, we get
$$A^1W_{x}\cdot W_{x}\le -\frac{
\theta_1}2|w^{I}_{x}|^2 + C|w^{II}_{x}|^2.$$ Therefore
\begin{align}\label{Ibd1}I_b^1 \le -\frac{\theta_1}2\int_{\RR^d_0}|w^{I}_x|^2+ \int_{\RR^d_0}\Big(|W|^2+|w^{II}_{x}|^2+|w^{II}_t|^2+
|w^{II}_{\tx\tx}|^2\Big).
\end{align}

Meanwhile, for the inflow case, since $A^1_{11}\ge \theta_1>0$, we have
$$|A^1W_{x}\cdot W_{x}|\le C|W_{x}|^2.$$
In this case, the invertibility of $A^1_{11}$ allows us to use
the hyperbolic equation to derive
\begin{align*}|w^{I}_{x_1}|&\le C(|w^{I}_t|
+ |w^{II}_{x}|+|w^{I}_{\tx}|).\end{align*}Therefore
we get \begin{align} I_b^1 \le
\int_{\RR^d_0}\Big(|W|^2+|W_t|^2+|w^{I}_{\tx}|^2+|w^{II}_x|^2+|w^{II}_{\tx\tx}|^2\Big).\end{align}

Now apply the standard Sobolev inequality
\begin{equation}|w(0)|^2 \le
C\|w\|_{L^2(\RR)}(\|w_{x}\|_{L^2(\RR)}+\|w\|_{L^2(\RR)})\end{equation}
to control the term $|w^{II}_{x_1}(0)|^2$ in $I_b^1$ in both cases. We get
\begin{equation}\label{Sob-ineq}\begin{aligned}&\int_{\RR^d_0}|w^{II}_{x_1}|^2
\le \epsilon'\|w^{II}_{xx}\|_0^2 + C\|w^{II}_{x}\|^2_0.\end{aligned}\end{equation}
Using this with $\epsilon'=\theta/8$, \eqref{Ib1}, and \eqref{Ibd1},
the estimate \eqref{ineq-key2} reads
\begin{equation}\label{ineq-key3}\begin{aligned}\dt
\wprod{A^0W_{x},W_{x}} &+\|w^{II}_{xx}\|^2_{0}+
\wprod{\omega(x_1)w^{I}_{x},w^{I}_{x} }\\&\le
C\Big(\zeta\|w^{I}\|_1^2 +\wprod{|\bW_{x_1}|w^{I},w^{I}}+
C(c_*)\|w^{II}\|_{1}^2\Big)+ I_b^1
\end{aligned}\end{equation}where the (new) boundary term $I_b^1$
satisfies
\begin{equation}\label{Ib1-out} I_b^1\le -\frac{\theta_1}2\int_{\RR^d_0}|w^{I}_x|^2+ C\int_{\RR^d_0}\Big(|W|^2+|w^{II}_{\tx}|^2+|w^{II}_t|^2+
|w^{II}_{\tx\tx}|^2\Big)\end{equation} for the outflow case, and
\begin{equation} \label{Ib1-in}I_b^1\le \int_{\RR^d_0}\Big(|W|^2+|W_t|^2+|W_{\tx}|^2+|w^{II}_{\tx\tx}|^2\Big)\end{equation}for the inflow case.

\subsubsection{Higher order ``Friedrichs-type'' estimate} For any fixed multi-index
$\alpha = (\alpha_{x_1},\cdots , \alpha_{x_d})$, $\alpha_1=0,1$,
$|\alpha|=k=2,...,s$, by computing $\dt \langle
A^0\partial_{x}^\alpha W,\partial_{x}^\alpha W\rangle$ and following
the same spirit as the above subsection, we easily obtain
\begin{equation}\label{F-higher1}\begin{aligned} \dt \langle
A^0&\partial_{x}^\alpha W,\partial_{x}^\alpha W\rangle +
\theta\|\partial^{\alpha+1}_{x}w^{II}\|_0^2+
\wprod{\omega(x_1)\partial^\alpha_{x}w^{I},\partial^{\alpha}_{x}w^{I}
}\\&\le C\Big(C(c_*)\|w^{II}\|_k^2 + \zeta\|w^{I}\|_k^2 +
\sum_{i=1}^{k-1}\wprod{|\bW_{x_1}|\partial_{x}^iw^{I},\partial_{x}^iw^{I}}\Big)+I_b^\alpha
\end{aligned}\end{equation} where \begin{align*}\partial_x^\alpha:&=\partial_{x_1}^{\alpha_1}\cdots\partial_{x_d}^{\alpha_d},\quad
\partial_x^{\alpha+1}:=\sum_j\partial_{x_1}^{\alpha_1}\cdots\partial_{x_d}^{\alpha_d} \partial_{x_j}, \quad \partial_x^i =\sum_{|\beta|=i}\partial_{x_1}^{\beta_1}\cdots\partial_{x_d}^{\beta_d}\end{align*}
 and the boundary term $I_b^\alpha$ satisfies
\begin{equation}\label{BCest-out} \begin{aligned}I_b^\alpha\le -\frac{\theta_1}2\int_{\RR^d_0}|\partial_x^\alpha w^{I}|^2+
C\int_{\RR^d_0}\Big(\sum_{i=1}^{[(k+1)/2]}|\partial_t^iw^{II}|^2+\sum_{i=0}^{k-1}|\partial^i_xw^{I}|^2+
\sum_{i=0}^k|\partial^i_{\tx}w^{II}|^2\Big)\end{aligned}\end{equation}
for the outflow case, and
\begin{equation} \label{BCest-in}I_b^\alpha\le \int_{\RR^d_0}\Big(\sum_{i=0}^k|\partial_t^iw^{I}|^2+\sum_{i=1}^{[(k+1)/2]}|\partial_t^iw^{II}|^2+
\sum_{i=0}^k|\partial^i_{\tx}W|^2\Big)\end{equation}for the inflow
case.

Now for $\alpha$ with $\alpha_1=2,...,s$ we observe that the
estimate \eqref{F-higher1} still holds. Indeed, using integration by
parts and computing $\dt \langle A^0\partial_{x}^\alpha
W,\partial_{x}^\alpha W\rangle$ as above
leaves the boundary terms
as
\begin{equation}\label{Ib-alpha}I_b^\alpha:=\frac12
\iprod{\partial^\alpha_x W,A^1\partial^\alpha_xW}
-\iprod{\partial^\alpha_xw^{II},\partial^\alpha_x[(b^{1k}w^{II}_{x_k})+(m_2^1\bW_{x_1})]}.\end{equation}
Then we can use the parabolic equations to solve $$w^{II}_{x_1x_1} = (b^{11})^{-1}\Big(A^0_2w^{II}_t+A^j_2W_{x_j} - (b^{jk}w^{II}_{x_k})_{x_j}
-b^{11}_{x_1}w^{II}_{x_1}- M_1\bW_{x_1} -
(m_2^j\bW_{x_1})_{x_j}\Big).$$
Using this we can reduce the order of derivative with respect to $x_1$ in $\partial_x^\alpha$ to one,
with the same spirit as \eqref{BCest2} and \eqref{BCest3}.
Finally we use the Sobolev embedding similar to \eqref{Sob-ineq} to obtain the
estimate for the normal derivative $\partial_{x_1}$,
 and get the estimate for $I_b^\alpha$ as
 claimed in \eqref{BCest-out} and \eqref{BCest-in}.

We recall next the following Kawashima-type estimate, presented in \cite{Z3}, to bound the term
$\|w^{I}\|_k^2$ appearing on the left hand side of \eqref{F-higher1}.

\subsubsection{``Kawashima-type'' estimate} Let $K(\xi)$ be the
skew-symmetry in \eqref{K1}. Using Plancherel's identity and the
equations \eqref{perturb-eqs}, we compute
\begin{equation}\label{K-est1} \begin{aligned} \frac 12 \dt\wprod{K(\partial_x)\partial^r_x W,
\partial_x^r}&= \frac 12 \dt \wprod{iK(\xi)(i\xi)^r\hat W,(i\xi)^r\hat
W}\\&= \wprod{iK(\xi)(i\xi)^r\hat W, (i\xi)^r\hat W_t}\\&= \wprod{
(i\xi)^r\hat W, -K(\xi)(A^0_+)^{-1}\sum_j\xi_jA^j_+(i\xi)^r\hat W}
\\&\quad + \wprod{iK(\xi)(i\xi)^r\hat W,(i\xi)^r\hat H},
\end{aligned}\end{equation}
where \begin{equation}\begin{aligned}
H:=&\sum_j\Big((A^0_+)^{-1}A^j_+ - (A^0)^{-1}A^j\Big)W_{x_j} \\&+
(A^0)^{-1}\Big(\sum_{jk}(B^{jk}W_{x_k})_{x_j} + M_1\bW_{x_1} +
\sum_j(M_2^j\bW_{x_1})_{x_j}\Big).
\end{aligned}\end{equation}

By using the fact that $|(A^0_+)^{-1}A^j_+ - (A^0)^{-1}A^j| =
\cO(\zeta + |\bW_{x_1}|)$, we can easily obtain
$$ \|\partial_x^rH\|_0^2 \le C\|w^{II}\|_{r+2}^2 + C\sum_{k=0}^{r+1}\wprod{(\zeta+|\bW_{x_1}|)\partial_x^{k}w^I,\partial_x^{k}w^I}.$$

Meanwhile, applying \eqref{K1} into the first term of the last line
in \eqref{K-est1}, we get
\begin{align*}\langle
(i\xi)^r\hat W, &-K(\xi)(A^0_+)^{-1}\sum_j\xi_jA^j_+(i\xi)^r\hat
W\rangle \\&\ge \theta\||\xi|^{r+1}\hat W\|_0^2 - C\||\xi|^{r+1}\hat
w^{II}\|_0^2\\&=\theta\|\partial_x^{r+1}w^{I}\|_0^2 -
C\|\partial_x^{r+1}w^{II}\|_0^2.\end{align*}

Putting these estimates together into \eqref{K-est1}, we have obtained
the high order ``Kawashima-type'' estimate:
\begin{equation}\label{K-est2}
\begin{aligned}
\dt\wprod{K(\partial_x)\partial_{x}^{r}W,\partial_{x}^{r}W}
\le&-\theta\|\partial_{x}^{r+1} w^{I}\|_0^2+
C\|w^{II}\|_{r+2}^2\\
&
+C\sum_{i=0}^{r+1}\wprod{(\zeta+|\bW_{x_1}|)\partial_{x}^{i}w^{I},\partial_{x}^{i}w^{I}}
\end{aligned}\end{equation}

\subsubsection{Final estimates} We are ready to conclude our result.
First combining the estimate \eqref{ineq-key3} with
\eqref{Fzeroth-est0}, we easily obtain
\begin{align*} \frac 12 \dt &\Big(\wprod{A^0W_{x},W_{x}} + M\wprod{A^0 W,W}\Big)
 \\\le& -\Big(\frac\theta8\|w^{II}_{xx}\|^2_{0}+\frac 14\wprod{\omega(x_1)w^{I}_{x},w^{I}_{x} }\Big)\\&+
C\Big(\zeta\|w^{I}\|_1^2  +\wprod{|\bW_{x_1}|w^{I},w^{I}}+
C(c_*)\|w^{II}\|_{1}^2\Big) + I_b^1
\\&-\frac M4\Big(\wprod{\omega(x_1) w^{I},w^{I}} +\theta
\|w^{II}_{x}\|_0^2\Big)+ CM\zeta\|w^{I}\|_0^2+ MC(c_*)\|w^{II}\|^2_0 + MI_b^0
\end{align*}

By choosing $M$ sufficiently large such that $M\theta \gg C(c_*)$,
and noting that $c_*\theta_1 |\bW_{x_1}|\le \omega(x_1)$, we get
\begin{equation}\label{F-combine01}\begin{aligned} \frac 12 \dt&
\Big(\wprod{A^0W_{x},W_{x}} + M\wprod{A^0 W,W}\Big)\\\le&
-\Big(\theta\|w^{II}\|^2_{2}+\wprod{\omega(x_1) w^{I},w^{I}}+
\wprod{\omega(x_1)w^{I}_{x},w^{I}_{x} }\Big)\\&+
C\Big(\zeta\|w^{I}\|_1^2+ C(c_*)\|w^{II}\|^2_0\Big) + I_b^1 +
MI_b^0.
\end{aligned}\end{equation}
We shall treat the boundary terms later.
Now we use the estimate
\eqref{K-est2} (for $r=0$) to absorb the term
$\|\partial_xw^{I}\|_0$ into the left hand side. Indeed, fixing
$c_*$ large as above, adding \eqref{F-combine01} with \eqref{K-est2}
times $\epsilon$, and choosing $\epsilon ,\zeta$ sufficiently small
such that $\epsilon C(c_*)\ll \theta, \epsilon \ll 1$ and $\zeta \ll
\epsilon \theta_2$, we obtain

\begin{align*} \frac 12 \dt \Big(&\wprod{A^0W_{x},W_{x}} +
M\wprod{A^0 W,W} + \epsilon \wprod{KW_{x},W} \Big) \notag\\\le&
-\Big(\theta\|w^{II}\|^2_{2}+\wprod{\omega(x_1) w^{I},w^{I}}+
\wprod{\omega(x_1)w^{I}_{x},w^{I}_{x} }\Big)\\&+
C\Big(\zeta\|w^{I}\|_1^2+ C(c_*)\|w^{II}\|^2_0\Big)
-\frac{\theta_2\epsilon}{2}\|w^{I}_{x}\|_0^2\\&+
C\epsilon\Big(\|w^{II}\|_{2}^2+\zeta\|w^{I}\|_0^2+\wprod{\omega(x_1)
w^{I},w^{I}}+ \wprod{\omega(x_1)w^{I}_{x},w^{I}_{x} }\Big)+ I_b^1
+ MI_b^0
\\\le &-\frac 12\Big(\theta\|w^{II}\|^2_{2} + \theta_2\epsilon\|w^{I}_{x}\|_0^2\Big)+
C(c_*)\|W\|_0^2+ I_b
\end{align*} where $I_b:=I_b^1 +
MI_b^0$.

In view of boundary terms $I_b^0$ and $I_b^1$, we treat the term $I_b$
in each inflow/outflow case separately. Recall the inequality
\eqref{Sob-ineq}, $\|w^{II}_{x_1}\|_{0,0}\le C\|w^{II}\|_2$. Thus,
using this, for the inflow case we have
\begin{equation}\begin{aligned}I_b^0&\le C(\|W\|_{0,0}^2
+ \|w^{II}_{x}\|_{0,0}\|w^{II}\|_{0,0})\le C(\|W\|_{0,0}^2 +
\|w^{II}_{\tx}\|_{0,0}^2+\epsilon \|w^{II}\|_2^2)
\end{aligned}\end{equation}
and for the outflow case,
\begin{equation}\begin{aligned}I_b^0&\le -\frac {\theta_1}{2}\|w^{I}\|_{0,0}^2 + C(\|w^{II}\|_{0,0}^2 +
\|w^{II}_{x}\|_{0,0}\|w^{II}\|_{0,0})\\&\le -\frac
{\theta_1}{2}\|w^{I}\|_{0,0}^2  + C(\|w^{II}\|_{0,0}^2 +
\|w^{II}_{\tx}\|_{0,0}^2+\epsilon
\|w^{II}\|_2^2).\end{aligned}\end{equation}

Therefore these together with \eqref{Ib1-out} and \eqref{Ib1-in},
using the good estimate of $\|w^{II}_{xx}\|_0^2$, yield
\begin{equation}\label{Ib-out} I_b\le
-\frac{\theta_1}2\int_{\RR^d_0}(|w^I|^2+|w^{I}_x|^2) +
C\int_{\RR^d_0}\Big(|w^{II}|^2+|w^{II}_{\tx}|^2+|w^{II}_t|^2+
|w^{II}_{\tx\tx}|^2\Big)\end{equation} for the outflow case, and
\begin{equation} \label{Ib-in}I_b^1\le \int_{\RR^d_0}\Big(|W|^2+|W_t|^2+|W_{\tx}|^2+|w^{II}_{\tx\tx}|^2\Big)\end{equation}
for the inflow case.

Now by Cauchy-Schwarz's inequality, $|K(\xi)|\le C|\xi|$, and
positive definiteness of $A^0$, it is easy to see that
\begin{equation}\begin{aligned}\cE:&=\wprod{A^0W_{x},W_{x}} + M\wprod{A^0 W,W} +
\epsilon \wprod{K(\partial_x)W,W} \sim \|W\|_{H^1_\alpha}^2 \sim
\|W\|_{H^1}^2.\end{aligned}\end{equation}  The last equivalence is
due to the fact that $\alpha$ is bounded above and below away from
zero.
Thus the above yields
\begin{align*}\dt \cE(W)(t)\le  - \theta_3 \cE(W)(t) + C(c_*)\Big(\|W(t)\|_{L^2}^2 +
|\CalB_1(t)|^2\Big),\end{align*} for some positive constant $\theta_3$,
which by the Gronwall inequality implies
\begin{equation}\label{energy-estimate}\|W(t)\|_{H^1}^2 \le Ce^{-\theta t}\|W_0\|_{H^1}^2 + C(c_*)\int_0^t e^{-\theta(t-\tau)}
\Big(\|W(\tau)\|_{L^2}^2 + |\CalB_1(\tau)|^2\Big)d\tau,\end{equation}
where $W(x,0)=W_0(x)$ and
\begin{equation}
|\CalB_1(\tau)|^2:= \int_{\RR^d_0}\Big(|W|^2+|W_t|^2+|W_{\tx}|^2+|w^{II}_{\tx\tx}|^2\Big)\end{equation} for the inflow case, and
\begin{equation}
|\CalB_1(\tau)|^2:= \int_{\RR^d_0}\Big(|w^{II}|^2+|w^{II}_{\tx}|^2+|w^{II}_t|^2+
|w^{II}_{\tx\tx}|^2\Big)\end{equation} for the outflow case.

Similarly, by induction, we can derive the same estimates for $W$ in
$H^s$. To do that, let us define
\begin{align*} \cE_1(W)&:= \wprod{A^0W_{x},W_{x}} + M\wprod{A^0 W,W} + \epsilon \wprod{KW_{x},W}
\\\cE_k(W)&:= \wprod{A^0\partial_{x}^k W,\partial_{x}^kW} + M\cE_{k-1}(W) + \epsilon \wprod{K\partial_{x}^kW,\partial_{x}^{k-1}W}
, \qquad k\le s
.\end{align*}

Then similarly by the Cauchy-Schwarz inequality,
$\cE_s(W) \sim
\|W\|_{H^s}^2$, and by induction, we obtain
\begin{align*}\dt \cE_s(W)(t)\le  - \theta_3 \cE_s(W)(t) + C(c_*)(\|W(t)\|_{L^2}^2+|\CalB_h(t)|^2),\end{align*} for some positive constant
$\theta_3$,
which by the Gronwall inequality yields
\begin{equation}\label{energy-estimate-higher}\|W(t)\|_{H^s}^2 \le Ce^{-\theta t}\|W_0\|_{H^s}^2 + C(c_*)\int_0^t
e^{-\theta(t-\tau)}(\|W(\tau)\|_{L^2}^2+|\CalB_h(\tau)|^2)d\tau,\end{equation}
where $W(x,0)=W_0(x)$,
and $\CalB_h$ are defined as in
\eqref{Bdry-out} and \eqref{Bdry-in}.

\subsubsection{The general case}
Following \cite{MaZ4,Z3}, the general
case that hypotheses (A1)-(A3) hold can easily be covered via
following simple observations. First, we may express matrix $A$ in
\eqref{perturb-eqs} as
\begin{equation}\label{formA}\begin{aligned}A^j(W+\bW) &= \hat A^j +
(\zeta + |\bW_{x_1}|)\begin{pmatrix}0 & \cO(1) \\
\cO(1) & \cO(1)\end{pmatrix},\end{aligned}\end{equation}
where $\hat A^j$ is a symmetric matrix obeying the same derivative
bounds as described for $A^j$, $\hat A^1$ identical to $A^1$ in the $11$ block and
obtained in other blocks $kl$ by
\begin{equation}\label{formAjk}\begin{aligned}A^1_{kl}(W+\bW) &=
A^1_{kl}(\bW)+A^1_{kl}(W+\bW) - A^1_{kl}(\bW)\\&= A^1_{kl}(W_+) +
\cO(|W_{x}|+|\bW_{x_1}|)\\&= A^1_{kl}(W_+) + \cO(\zeta+|\bW_{x_1}|)
\end{aligned}\end{equation}
and meanwhile, $\hat A^j$, $j\not=1$, obtained by $A^j= A^j(W_+) + \cO(\zeta+|\bW_{x_1}|)$,
similarly as in \eqref{formAjk}.

Replacing $A^j$ by $\hat A^j$ in the $k^{th}$ order
Friedrichs-type bounds above, we find that the resulting error terms
may be expressed as
$$\wprod{\partial_{x}^k\cO(\zeta +
|\bW_{x_1}|)|W|,|\partial_{x}^{k+1}w^{II}|},$$ plus lower order
terms, easily absorbed using Young's inequality, and boundary terms
$$\cO(\sum_{i=0}^k|\partial_{x}^iw^{II}(0)||\partial^{k}_{x}w^{I}(0)|)$$ resulting from the use of integration by
parts as we deal with the $12$-block. However these boundary terms
were already treated somewhere as before.
Hence we can recover the same Friedrichs-type estimates obtained
above. Thus we may relax $(A1')$ to $(A1)$.

Next, to relax $(A3')$ to $(A3)$, first we show that the symmetry condition $B^{jk}=B^{kj}$ is not necessary. Indeed, by writing \begin{align*}\sum_{jk}(B^{jk}W_{x_k})_{x_j} = \sum_{jk}\Big(\frac12(B^{jk}+B^{kj})W_{x_k}\Big)_{x_j} + \frac12\sum_{jk}(B^{jk}-B^{kj})_{x_j}W_{x_k},\end{align*}
we can just replace $B^{jk}$ by $\tilde B^{jk}:= \frac12(B^{jk}+B^{kj})$, satisfying the same $(A3')$, and thus still obtain the energy estimates as before, with a harmless error term (last term in the above identity).
Next notice that the term $g(\tilde W_{x})-g(\bW_{x_1})$
in the perturbation equation may be Taylor expanded as
$$\begin{pmatrix}0 \\
g_1(\tilde W_{x},\bW_{x_1})+g_1(\bW_{x_1},\tilde W_{x})\end{pmatrix}+\begin{pmatrix}0 \\
\cO(|W_{x}|^2)\end{pmatrix}$$
The first term, since it vanishes in the first component and since
$|\bar W_x|$ decays at plus spatial infinity,
yields by Young's inequality the estimate
$$
\Big\langle
\begin{pmatrix}0 \\
g_1(\tilde W_{x},\bW_{x_1})+g_1(\bW_{x_1},\tilde W_{x})\end{pmatrix},
\begin{pmatrix}w^{I}_{x} \\
w^{II}_{x}\end{pmatrix} \Big\rangle
\le
C\Big(\wprod{(\zeta+|\bW_{x_1}|)w^{I}_{x},w^{I}_{x}}
+\|w^{II}_{x}\|_{0}^2\Big)$$
which can be treated in the
Friedrichs-type estimates. The $(0,O(|W_{x}|^2)^T$ nonlinear term may
be treated as other source terms in the energy estimates.
Specifically, the
worst-case term \begin{align*}\Big<\partial_{x}^k W,&\partial_{x}^k\begin{pmatrix}0 \\
\cO(|W_{x}|^2)\end{pmatrix}\Big> \\&= -\wprod{\partial_{x}^{k+1}
w^{II},\partial_{x}^{k-1}\cO(|W_{x}|^2)}-\partial_{x}^k
w^{II}(0)\partial_{x}^{k-1}\cO(|W_{x}|^2)(0)\end{align*} may be bounded by
$$\|\partial_{x}^{k+1}
w^{II}\|_{L^2}\|W\|_{W^{2,\infty}}\|W\|_{H^k} -\partial_{x}^k
w^{II}(0)\partial_{x}^{k-1}\cO(|W_{x}|^2)(0).$$ The boundary term
will contribute to energy estimates
in the form \eqref{Ib-alpha} of $I_b^\alpha$, and
thus we may use the parabolic equations to get rid of this term as
we did in \eqref{BCest2}, \eqref{BCest3}.
Thus, we may relax $(A3')$ to $(A3)$,
completing the proof of the general case $(A1)-(A3)$ and the
proposition.\end{proof}

\subsection{Proof of nonlinear stability}
Defining the perturbation variable $U:= \tilde U - \bU$, we obtain
the nonlinear perturbation equations
\begin{equation}\label{per-eqs} U_t - LU = \sum_j
Q^j(U,U_x)_{x_j},\end{equation} where
\begin{equation}\label{newqbounds}
\begin{aligned}
Q^j(U,U_x)&=\cO(|U||U_x|+|U|^2)\\
Q^j(U,U_x)_{x_j}&= \cO(|U||U_{x}|+|U||U_{xx}|+|U_x|^2)\\
Q^j(U,U_x)_{x_jx_k}&=
\cO(|U||U_{xx}|+|U_x||U_{xx}|+|U_x|^2+|U||U_{xxx}|)
\end{aligned}
\end{equation}
so long as $|U|$ remains bounded.

For boundary conditions written in $U$-coordinates, (B) gives
\begin{equation}\label{BCs-in}
\begin{aligned}
h=\tilde h -\bar h&=
(\tilde W(U+\bar U)-\tilde W(\bar U))(0,\tx,t)\\
&=(\partial \tilde W/\partial \tilde U)(\bar U_0) U (0,\tx,t)+
\cO(|U(0,\tx,t)|^2).
\end{aligned}
\end{equation}
in inflow case, where $(\partial \tilde W/\partial \tilde U)(\bar U_0)$
is constant and invertible, and
\begin{equation}\label{BCs-out}
\begin{aligned}
h=\tilde h -\bar h&=
(\tilde w^{II}(U+\bar U)-\tilde w^{II}(\bar U))(0,\tx,t)\\
&= (\partial \tilde w^{II}/\partial \tilde U)(\bar U_0)
U(0,\tx,t) + \cO(|U(0,\tx,t)|^2)\\
&= m\begin{pmatrix} \bar b_1 & \bar b_2 \end{pmatrix}(\bar U_0)
U(0,\tx,t) + \cO(|U(0,\tx,t)|^2)
\\&= m B(\bar U_0)U(0,\tx,t) + \cO(|U(0,\tx,t)|^2)
\end{aligned}
\end{equation}for some invertible constant matrix $m$.

Applying Lemma \ref{lem-duhamel} to \eqref{per-eqs}, we obtain
\begin{equation}\label{Duhamel}
\begin{aligned}
  U(x,t)=& \cS(t) U_0 + \int_0^t \cS(t-s)\sum_j
\partial_{x_j}Q^j(U,U_x)ds + \Gamma U(0,\tx,t)
\end{aligned}
\end{equation} where $U(x,0) = U_0(x),$ \begin{equation}\label{GammaU}\Gamma U(0,\tx,t):=\int_0^t\int_{\RR^{d-1}}(\sum_jG_{y_j}B^{j1}+GA^1)(x,t-s;0,\tilde y)U(0,\tilde y,s)\,d\tilde yds,\end{equation} and
$G$ is the Green function of $\partial_t-L$.

\begin{proof}[Proof of Theorem \ref{theo-nonlin}]
Define \begin{equation}\label{zeta} \begin{aligned}\zeta(t):=\sup_s
&\Big(|U(s)|_{L^2_x}(1+s)^{\frac{d-1}4}+|U(s)|_{L^\infty_x}(1+s)^{\frac
d2}
\\&+(|U(s)|+|U_x(s)|+|\partial_{\tx}^2U(s)|)_{L^{2,\infty}_{\tx,x_1}}(1+s)^{\frac{d+1}4}
\Big).\end{aligned}
\end{equation}

 We shall prove here that for all $t\ge
0$ for which a solution exists with $\zeta(t)$ uniformly bounded by
some fixed, sufficiently small constant, there holds
\begin{equation}\label{zeta-est}
\zeta(t) \le C(|U_0|_{L^1\cap H^s}+E_0+\zeta(t)^2) .\end{equation}

This bound together with continuity of $\zeta(t)$ implies that
\begin{equation}\label{zeta-est1} \zeta(t) \le 2C(|U_0|_{L^1\cap H^s}+E_0)\end{equation}
for $t\ge0$, provided that $|U_0|_{L^1\cap H^s} +E_0< 1/4C^2$. This
would complete the proof of the bounds as claimed in the theorem,
and thus give the main theorem.

By standard short-time theory/local well-posedness in $H^s$, and the
standard principle of continuation, there exists a solution $U\in
H^s$ on the open time-interval for which $|U|_{H^s}$ remains
bounded, and on this interval $\zeta(t)$ is well-defined and
continuous. Now, let $[0,T)$ be the maximal interval on which
$|U|_{H^s}$ remains strictly bounded by some fixed, sufficiently
small constant $\delta>0$. By Proposition \ref{prop-energy-est}, and
the Sobolev embeding inequality $|U|_{W^{2,\infty}}\le C|U|_{H^s}$,
we have
\begin{equation}\label{Hs}\begin{aligned}|U(t)|_{H^s}^2 &\le Ce^{-\theta t}|U_0|_{H^s}^2
+ C \int_0^t
e^{-\theta(t-\tau)}\Big(|U(\tau)|_{L^2}^2+|\CalB_h(\tau)|^2\Big)d\tau\\&\le
C(|U_0|_{H^s}^2+ E_0^2 +\zeta(t)^2)(1+t)^{-(d-1)/2}.
\end{aligned}\end{equation}
and so the solution continues so long as $\zeta$ remains small, with
bound \eqref{zeta-est1}, yielding existence and the claimed bounds.

Thus, it remains to prove the claim \eqref{zeta-est}. First by \eqref{Duhamel}, we obtain
\begin{equation}\begin{aligned} |U(t)|_{L^2}\le& |\cS(t)U_0|_{L^2} +
\int_0^t|\cS_1(t-s)\partial_{x_j}Q^j(s)|_{L^2}ds \\&+ \int_0^t
|\cS_2(t-s)\partial_{x_j}Q^j(s)|_{L^2}ds +|\Gamma  U(0,\tx,t)|_{L^2}\\\le& I_1 + I_2+I_3+|\Gamma  U(0,\tx,t)|_{L^2}
\end{aligned}\end{equation}

where
$$\begin{aligned}I_1:&=|\cS(t) U_0|_{L^2}\le C (1+t)^{-\frac{d-1}{4}}|U_0|_{L^1\cap
H^3},\\
I_2:&=\int_0^t|\cS_1(t-s)\partial_{x_j}Q^j(s)|_{L^2}ds
\\&\le C\int_0^t (1+t-s)^{-\frac{d-1}{4}-\frac12}|Q^j(s)|_{L^1} + (1+s)^{-\frac{d-1}4}|Q^j(s)|_{L^{1,\infty}_{\tx,x_1}} ds
\\&\le C\int_0^t (1+t-s)^{-\frac{d-1}{4}-\frac12}|U|_{H^1}^2
+ (1+t-s)^{-\frac{d-1}4}\Big(|U|^2_{L^{2,\infty}_{\tx,x_1}}+|U_x|^2_{L^{2,\infty}_{\tx,x_1}}\Big)ds\\&\le C(|U_0|_{H^s}^2+\zeta(t)^2)\int_0^t
\Big[ (1+t-s)^{-\frac{d-1}{4}-\frac12}(1+s)^{-\frac{d-1}{2}}
\\&\quad\qquad+ (1+t-s)^{-\frac{d-1}4}(1+s)^{-\frac{d+1}2}\Big]ds\\&\le
C(1+t)^{-\frac{d-1}{4}}(|U_0|_{H^s}^2+\zeta(t)^2)
\end{aligned}$$
and
$$\begin{aligned} I_3:&=\int_0^t
|\cS_2(t-s)\partial_{x_j}Q^j(s)|_{L^2}ds\\&\le \int_0^t
e^{-\theta(t-s)}|\partial_{x_j}Q^j(s)|_{H^3}ds
\\&\le C\int_0^t
e^{-\theta(t-s)}(|U|_{L^\infty} + |U_x|_{L^\infty})|U|_{H^5}ds
\\&\le C\int_0^t
e^{-\theta(t-s)}|U|_{H^s}^2ds
\\&\le C(|U_0|_{H^s}^2+\zeta(t)^2) \int_0^t
e^{-\theta(t-s)}(1+s)^{-\frac{d-1}{2}}ds\\&\le
C(1+t)^{-\frac{d-1}{2}}(|U_0|_{H^s}^2+\zeta(t)^2).
\end{aligned}$$
Meanwhile, for the boundary term $|\Gamma  U(0,\tx,t)|_{L^2}$, we treat two cases separately. First for the inflow case, then by \eqref{BCs-in} we have $$|U(0,\tx,t)|\le C|h(\tx,t)| + \cO(|U(0,\tx,t)|^2),$$ and thus $|U(0,\tx,t)|\le C|h(\tx,t)|$, provided that $|h|$ is sufficiently small. Therefore under the
hypotheses on $h$ in Theorem \ref{theo-nonlin}, Proposition \ref{prop-BCs-est} yields $$|\Gamma  U (0,\cdot,\cdot)|_{L^2_x} \le CE_0 (1+t)^{-\frac{d-1}{4}}.$$

Now for the outflow case, recall that in this case $G(x,t;0,\tilde y) \equiv 0$. Thus \eqref{GammaU} simplifies to \begin{equation}\label{cBUest}\Gamma  U(0,\tx,t)=\int_0^t\int_{\RR^{d-1}}G_{y_1}(x,t-s;0,\tilde y)B^{11}U(0,\tilde y,s)\,d\tilde yds.\end{equation}
To deal with this term, we shall use Proposition \ref{prop-BCs-est} as in the inflow case.
In view of \eqref{BCs-out},
$$|B^{11}U(0,\tilde y,s)| \le C|h(\tilde y,t)|+\cO(|U(0,\tilde y,s)|^2),$$ and assumptions on $h$ are imposed as in Theorem \ref{theo-lin},
so that \eqref{hdecay} is satisfied.
To check the last term $\cO(|U(0)|^2),$ using the definition \eqref{zeta} of $\zeta(t)$, we have
$$\begin{aligned} |\cO(|U(0,\tilde y,s)|^2)|_{L^2}&\le C|U|_{L^{\infty}}|U|_{L^{2,\infty}_{\tilde x,x_1}} \le C\zeta^2(t)(1+s)^{-\frac d2 - \frac{d+1}{4}}\\|\cO(|U(0,\tilde y,s)|^2)|_{L^{\infty}}&\le C|U|_{L^{\infty}}^2\le C\zeta^2(t)(1+s)^{-d}\end{aligned} $$ and for the term $\cD_h$ with $h$ replaced by $\cO(|U(0,\tilde y,s)|^2)$, using the standard H{\"{o}}lder inequality to get
$$\begin{aligned} |\cD_h|_{L^1_{\tx}}&\le C(|U|_{L^{2,\infty}}^2+|U_x|^2_{L^{2,\infty}}+|U_{\tx\tx}|^2_{L^{2,\infty}})\le C\zeta^2(t)(1+s)^{-\frac {d+1}{2}}\\
|\cD_h|_{H^{[(d-1)/2]+5}_{\tx}}
&\le C|U|_{L^{\infty}}|U|_{H^s}\le C\zeta^2(t)(1+s)^{-d/2-(d-1)/4}.\end{aligned} $$
We remark here that Sobolev bounds \eqref{Hs} are not good enough for estimates of $\cD_h$ in $L^1$, requiring a decay at rate $(1+t)^{-d/2-\epsilon}$ for the two-dimensional case (see Proposition \ref{prop-BCs-est}). This is exactly why we have to keep track of $U_{\tx\tx}$ in $L^{2,\infty}$ norm in $\zeta(t)$ as well, to gain a bound as above for $\cD_h$.

Therefore applying Proposition \ref{prop-BCs-est}, we also obtain \eqref{cBUest} for the outflow case. Combining these above estimates yields \begin{equation} |U(t)|_{L^2}(1+t)^{\frac{d-1}{4}} \le
C(|U_0|_{L^1\cap H^s}+E_0+\zeta(t)^2) .\end{equation}

Next, we estimate \begin{equation}\label{estU}\begin{aligned}
|U(t)|_{L^{2,\infty}_{\tx,x_1}}\le&
|\cS(t)U_0|_{L^{2,\infty}_{\tx,x_1}} +
\int_0^t|\cS_1(t-s)\partial_{x_j}Q^j(s)|_{L^{2,\infty}_{\tx,x_1}}ds
\\&+ \int_0^t |\cS_2(t-s)\partial_{x_j}Q^j(s)|_{L^{2,\infty}_{\tx,x_1}}ds +|\Gamma  U(0,\tx,t)|_{L^{2,\infty}_{\tx,x_1}}\\\le& J_1 +
J_2+J_3+|\Gamma  U(0,\tx,t)|_{L^{2,\infty}_{\tx,x_1}}
\end{aligned}\end{equation}
where
$$\begin{aligned}
J_1:&=|\cS(t)U_0|_{L^{2,\infty}_{\tx,x_1}}\le C
(1+t)^{-\frac{d+1}{4}}|U_0|_{L^1\cap
H^{4}} \\
J_2:&=\int_0^t|\cS_1(t-s)\partial_{x_j}Q^j(s)|_{L^{2,\infty}_{\tx,x_1}}ds
\\&\le C\int_0^t (1+t-s)^{-\frac{d+1}{4}-\frac12}|Q^j(s)|_{L^1} + (1+s)^{-\frac{d+1}4}|Q^j(s)|_{L^{1,\infty}_{\tx,x_1}} ds
\\&\le C\int_0^t (1+t-s)^{-\frac{d+1}{4}-\frac12}|U|_{H^1}^2
+ (1+t-s)^{-\frac{d+1}4}\Big(|U|^2_{L^{2,\infty}_{\tx,x_1}}+|U_x|^2_{L^{2,\infty}_{\tx,x_1}}\Big)ds
\\&\le C(|U_0|_{H^s}^2+\zeta(t)^2)\int_0^t
(1+t-s)^{-\frac{d+1}{4}-\frac12}(1+s)^{-\frac{d-1}{2}} \\&\quad\qquad+
(1+t-s)^{-\frac{d+1}4}(1+s)^{-\frac{d+1}2}ds
\\&\le
C(1+t)^{-\frac{d+1}{4}}(|U_0|_{H^s}^2+\zeta(t)^2)
\end{aligned}$$
and (by Moser's inequality)
$$\begin{aligned}J_3:&=\int_0^t|\cS_2(t-s)\partial_{x_j}Q^j(s)|_{L^{2,\infty}_{\tx,x_1}}ds
\\&\le C\int_0^t e^{-\theta(t-s)}|\partial_{x_j}Q^j(s)|_{H^4}ds
\\&\le C\int_0^t e^{-\theta(t-s)}|U|_{L^{\infty}_{x}}|U|_{H^6}ds
\\&\le C(|U_0|_{H^s}^2+\zeta(t)^2)\int_0^t
e^{-\theta(t-s)}(1+s)^{-\frac{d}{2}}(1+s)^{-\frac{d-1}{4}}ds
\\&\le
C(1+t)^{-\frac{d+1}{4}}(|U_0|_{H^s}^2+\zeta(t)^2).
\end{aligned}$$
These estimates together with similar treatment for the boundary term yield \begin{equation}
|U(t)|_{L^{2,\infty}_{\tx,x_1}}(1+t)^{\frac{d+1}{4}} \le
C(|U_0|_{L^1\cap H^s}+E_0+\zeta(t)^2) .\end{equation}

Similarly, we have the same estimate for
$|U_x(t)|_{L^{2,\infty}_{\tx,x_1}}$. Indeed, we have
\begin{equation}\begin{aligned}
|U_x(t)|_{L^{2,\infty}_{\tx,x_1}}\le&
|\partial_x\cS(t)U_0|_{L^{2,\infty}_{\tx,x_1}} +
\int_0^t|\partial_x\cS_1(t-s)\partial_{x_j}Q^j(s)|_{L^{2,\infty}_{\tx,x_1}}ds
\\&+ \int_0^t |\partial_x\cS_2(t-s)\partial_{x_j}Q^j(s)|_{L^{2,\infty}_{\tx,x_1}}ds +|\partial_x\Gamma  U(0,\tx,t)|_{L^{2,\infty}_{\tx,x_1}}\\\le& K_1 +
K_2+K_3+|\partial_x\Gamma  U(0,\tx,t)|_{L^{2,\infty}_{\tx,x_1}}
\end{aligned}\end{equation}
where $K_2$ and $K_3$ are treated exactly in the same way as the
treatment of $J_2,J_3$, yet in the first term of $K_2$ it is a bit
better by a factor
$t^{-1/2}$. Similar bounds hold for $|U_{\tx\tx}|$ in $L^{2,\infty}$, noting that there are no higher derivatives
 in $x_1$ involved and thus similar to those in \eqref{estU}.

Finally, we estimate the $L^\infty$ norm of $U$. By Duhamel's
formula \eqref{Duhamel}, we obtain
\begin{equation}\begin{aligned} |U(t)|_{L^\infty}\le& |\cS(t)U_0|_{L^\infty} +
\int_0^t|\cS_1(t-s)\partial_{x_j}Q^j(s)|_{L^\infty}ds
\\&+ \int_0^t |\cS_2(t-s)\partial_{x_j}Q^j(s)|_{L^\infty}ds+|\Gamma  U(0,\tx,t)|_{L^\infty} \\\le& L_1 +
L_2+L_3 +|\Gamma  U(0,\tx,t)|_{L^\infty}
\end{aligned}\end{equation}
where the boundary term is treated in the same way as above, and for 
$|\gamma| =[(d-1)/2]+2$,
$$\begin{aligned}L_1:&=|\cS(t) U_0|_{L^\infty}\le C (1+t)^{-\frac{d}{2}}|U_0|_{L^1\cap
H^{|\gamma|+3}},\\
L_2:&=\int_0^t|\cS_1(t-s)\partial_{x_j}Q^j(s)|_{L^\infty}ds
\\&\le C\int_0^t (1+t-s)^{-\frac{d}{2}-\frac12}|Q^j(s)|_{L^1} + (1+s)^{-\frac{d}2}|Q^j(s)|_{L^{1,\infty}_{\tx,x_1}} ds
\\&\le C\int_0^t (1+t-s)^{-\frac{d}{2}-\frac12}|U|_{H^1}^2
+ (1+t-s)^{-\frac{d}2}\Big(|U|^2_{L^{2,\infty}_{\tx,x_1}}+|U_x|^2_{L^{2,\infty}_{\tx,x_1}}\Big)ds\\&\le C(|U_0|_{H^s}^2+\zeta(t)^2)\int_0^t
\Big[ (1+t-s)^{-\frac{d}{2}-\frac12}(1+s)^{-\frac{d-1}{2}}
\\&\qquad\quad+ (1+t-s)^{-\frac{d}2}(1+s)^{-\frac{d+1}2}\Big]ds \\&\le
C(1+t)^{-\frac{d}{2}}(|U_0|_{H^s}^2+\zeta(t)^2)
\end{aligned}$$
and (again by Moser's inequality),
$$\begin{aligned} L_3:&=\int_0^t
|\cS_2(t-s)\partial_{x_j}Q^j(s)|_{L^\infty}ds\\&\le\int_0^t
|\cS_2(t-s)\partial_{x_j}Q^j(s)|_{H^{|\gamma|}}ds\\&\le \int_0^t
e^{-\theta(t-s)}|\partial_{x}Q^j(s)|_{H^{|\gamma|+3}}ds
\\&\le C\int_0^t
e^{-\theta(t-s)}|U|_{L^\infty}|U|_{H^{|\gamma|+5}}ds
\\&\le C(|U_0|_{H^s}^2+\zeta(t)^2) \int_0^t
e^{-\theta(t-s)}(1+s)^{-\frac{d}{2}}(1+s)^{-\frac{d-1}{4}}ds \\&\le
C(1+t)^{-\frac{d}{2}}(|U_0|_{H^s}^2+\zeta(t)^2).
\end{aligned}$$

Therefore we have obtained\begin{equation}
|U(t)|_{L^{\infty}_{x}}(1+t)^{\frac{d}{2}} \le C(|U_0|_{L^1\cap
H^s}+E_0+\zeta(t)^2)\end{equation} and thus completed the proof of claim
\eqref{zeta-est}, and the main theorem.\end{proof}


\appendix


\section{Physical discussion in the isentropic case}\label{phys}

In this appendix, we revisit in slightly more detail the drag-reduction
problem sketched in Examples \ref{aeroexam}--\ref{aeroexamin},
in the simplified context of the two-dimensional isentropic case.
Following the notation of \cite{GMWZ5},
consider the two-dimensional isentropic compressible 
Navier--Stokes equations
\begin{align}
\rho_t+ (\rho u)_x +(\rho v)_y&=0,\label{eq:mass}\\
(\rho u)_t+ (\rho u^2)_x + (\rho uv)_y + p_x&=
(2\mu +\eta) u_{xx}+ \mu u_{yy} +(\mu+ \eta )v_{xy},\label{eq:momentumx}\\
(\rho v)_t+ (\rho uv)_x + (\rho v^2)_y + p_y&= \mu v_{xx}+ (2\mu
+\eta) v_{yy} + (\mu+\eta) u_{yx}\label{eq:momentumy}
\end{align}
on the half-space $y>0$, where $\rho$ is density, $u$ and $v$ are
velocities in $x$ and $y$ directions, and $p=p(\rho)$ is pressure,
and $\mu>|\eta|\ge0$ are coefficients of first (``dynamic'') and
second viscosity, making the standard monotone pressure
assumption $p'(\rho)>0$.

We imagine a porous airfoil lying along the $x$-axis, with constant
imposed normal velocity $v(0)=V$ and zero transverse relative velocity
$u(0)=0$ imposed at the airfoil surface, and seek a laminar
boundary-layer flow $(\rho, u,v)(y)$ with transverse
 relative velocity $u_\infty$ a short distance away the airfoil, with
$|V|$ much less than the sound speed $c_\infty$ and $|u_\infty|$ of an
order roughly comparable to $c_\infty$.

\subsection{Existence}
The possible boundary-layer solutions have been completely categorized 
in this case in Section 5.1 of \cite{GMWZ5}.
We here cite the relevant conclusions, referring to \cite{GMWZ5}
for the (straightforward) justifying computations.

\subsubsection{Outflow case ($V<0$)} In the outflow case, the scenario
described above corresponds to case (5.15) of \cite{GMWZ5},
in which it is found that the only solutions are purely {\it transverse}
flows
\begin{equation}\label{transverse}
(\rho, v)\equiv (\rho_0, V),\quad
u(y)= u_\infty(1- e^{\rho_0 V y/\mu}),
\end{equation}
varying only in the tranverse velocity $u$.
The drag force per unit length at the airfoil,
by Newton's law of viscosity, is
\begin{equation}\label{drag}
\mu \bar u_{y}|_{y=0}= u_\infty \rho_\infty |V|,
\end{equation}
since momentum $m:=\rho_0 V=\rho_\infty V$ is constant throughout
the layer, so that ($\rho_\infty$, $u_\infty$ being
imposed by ambient conditions away from the wing) 
{\it drag is proportional to the speed $|V|$ of the imposed normal velocity}.



\subsubsection{Inflow case ($V>0$)}
Consulting again \cite{GMWZ5} (p. 61), we find for $V>0$ with
specified $(\rho, u,v)(0)$ of the orders described above, 
the only solutions are purely {\it normal} flows,
\begin{equation}\label{normal}
u \equiv u(0),\quad (\rho,v)=(\rho, v)(y),
\end{equation}
varying only in the normal velocity $v$.
Thus, it is not possible to reconcile the velocity $u(0)$ at the airfoil
with the velocity $u_\infty >>c$ some distance away.

As discussed in \cite{MN}, the expected behavior in such a case consists
rather of a combination of a boundary-layer at $y=0$
and one or more elementary planar shock, rarefaction, or contact waves 
moving away from $y=0$: in this case a shear wave moving with normal
fluid velocity $V$ into the half-space, across which the transverse
velocity changes from zero to $u_\infty$.
That is, a characteristic layer analogous to the solid-boundary
case {\it detaches} from the airfoil and travels outward into the
flow field.
In this case, one would not expect drag reduction compared to
the solid-boundary case, but rather some increase.

\subsection{Stability}
If we consider one-dimensional stability, or stability
with respect to perturbations depending only on $y$, 
we find that the linearized eigenvalue equations decouple
into the constant-coefficient linearized eigenvalue equations for $(\rho, v)$
about a constant layer $(\rho, v)\equiv (\rho_0, V)$,
and the scalar linearized eigenvalue equation
\begin{equation}\label{c-d}
\lambda \bar \rho u+m u_y=\mu u_{yy}
\end{equation}
associated with the constant-coefficient convection-diffusion equation
$ \bar \rho u_t +m u=\mu u_{yy},  $
$m:=\bar \rho \bar v\equiv \rho_0 V$, $\bar \rho\equiv \rho_0$.
As the constant layer $(\rho_0, V)$
is stable by Corollary \ref{specstab} or direct
calculation (Fourier transform), and
\eqref{c-d} is stable by direct calculation,
we may thus conclude that purely transverse layers are {\it one-dimensionally
stable}.

Considered with respect to general perturbations, the equations do
not decouple, nor do they reduce to constant-coefficient form,
but to a second order system whose coefficients are quadratic polynomials
in $e^{\rho_0 Vy}$.
It would be very interesting to try to resolve the question of
spectral stability by direct solution using this special form,
or, alternatively, to perform a numerical study as done in \cite{HLyZ2}
for the multi-dimensional shock wave case.

\begin{remark}\label{gen1d}
\textup{
For general laminar boundary layers $(\bar \rho, \bar u, \bar v)(y)$,
the one-dimensional stability problem, now variable-coefficient,
does not completely decouple, but has triangular form, breaking
into a system in $(\rho, v)$ alone and an equation in $u$ forced
by $(\rho,v)$.
Stability with respect to general perturbations, therefore, 
is equivalent to stability with respect to perturbations of form 
$(\rho, 0, v)$ or $(0,u,0)$.
For perturbations $(\rho,u,v)=(0,u,0)$, the $u$ equation again becomes
\eqref{c-d}, with $\mu$, $m$ still constant, but $\bar \rho$ varying
in $y$.  Taking the real part of the complex $L^2$ inner product 
of $u$ against \eqref{c-d} gives
$$
\Re \lambda \|u\|_{L^2}^2+\|u_y\|_{L^2}^2=0,
$$
hence for $\Re \lambda \ge 0$, $u\equiv \const=0$.
Thus, the layer is one-dimensionally stable if and only if 
the normal part $(\bar \rho,\bar v)$
is stable with respect to perturbations $(\rho, v)$.
Stability of normal layers was studied in \cite{CHNZ} for
a $\gamma$-law gas $p(\rho)=a\rho^\gamma$, $1\le \gamma \le 3$,
with the conclusion that {\it all layers are one-dimensionally stable}, 
independent of amplitude, in the general inflow and compressive outflow cases.
Hence, we can make the same conclusion for full layers 
$(\bar \rho, \bar u,\bar v)$.
In the present context, this includes all cases except for suction
with supersonic velocity $|V|>c_\infty$, which in the 
notation of \cite{CHNZ} is of {\it expansive outflow} type
(expected also to be stable, but not considered in \cite{CHNZ}), 
since $|\bar v|$ is decreasing
with $y$, so that density $\bar \rho$ (since $m=\bar \rho \bar v\equiv \const$)
is increasing.
}
\end{remark}

\subsection{Discussion}
Note that we do not achieve by subsonic boundary suction
an exact laminar flow connecting the values $(u,v)=(0,V)$
at the wing to the values $(u_\infty,0)$ of the ambient
flow at infinity, but rather to an intermediate value $(u_\infty, V)$.
That is, we trade a large variation $u_\infty$ in shear for a
possibly small variation $V$ in normal velocity, which appears
now as a boundary condition for the outer, approximately Euler flow
away from the boundary layer.
Whether the full solution is stable appears to be
a question concerning also nonstationary Euler flow.
It is not clear either what is the optimal outflux velocity $V$.
From \eqref{drag} and the discussion just above, it appears desirable
to minimize $|V|$, since this minimizes both drag and 
the imbalance between flow $v_\infty$ just outside the boundary layer and
the ambient flow at infinity. 
On the other hand, we expect that stability becomes more delicate
in the characteristic limit $V\to 0^-$, in the sense that the
size of the basin of attraction of the boundary layer shrinks to zero
(recall, we have ignored throughout our analysis the size 
of the basin of attraction, taking perturbations as small
as needed without keeping track of constants).
These would be quite interesting issues for further investigation.




\end{document}